\documentclass[final]{dmtcs-episciences}

\publicationdetails{22}{2020}{4}{8}{5583}
\usepackage[a4paper]{geometry}
\usepackage{amsmath}
\usepackage{amssymb}
\usepackage{amsthm}
\usepackage{bm}
\usepackage[shortlabels]{enumitem}
\usepackage{framed}
\usepackage{color}
\usepackage{xspace}
\usepackage{comment}

\usepackage[textwidth=20mm,obeyFinal,colorinlistoftodos]{todonotes}

\usepackage{graphicx}

\usepackage[utf8]{inputenc}

\usepackage{microtype}

\theoremstyle{plain}
\newtheorem{theorem}{Theorem}
\newtheorem{corollary}{Corollary}
\newtheorem{lemma}{Lemma}
\newtheorem{proposition}{Proposition}

\theoremstyle{definition}
\newtheorem{definition}{Definition}

\theoremstyle{remark}

\newcommand{\Oh}{\mathcal{O}}

\def\AA{\mathcal{A}}

\def\CC{\mathcal{C}}
\def\CCtw{\mathcal{C}^{\textrm{tw}}}

\def\FF{\mathcal{F}}

\def\WW{\mathcal{W}}
\def\NN{\mathbb{N}}
\def\RR{\mathbb{R}}
\def\ZZ{\mathbb{Z}}
\def\EF{Ehrenfeucht-Fraïssé }

\def\MSO{\textrm{MSO\xspace{}}}
\def\MSOk{\rm MSO[{\it k}]}

\def\MSOkeq{\equiv^{MSO}_k}

\def\equivMSOkt{\equiv^{MSO}_{k}}
\def\equivMSOktplus{\equiv^{MSO}_{k}}

\newcommand{\conv}{\mathop{\mathrm{conv}}}

\newcommand{\suppo}{\mathop{\mathrm{supp}}} 
\newcommand{\vertexset}{\mathop{\mathrm{vert}}}
\newcommand{\xc}{\mathop{\mathrm{xc}}} 
\newcommand{\xcdec}{\xc_{dec}} 

\begin{document}

\title{Extension Complexity, MSO Logic, and Treewidth\thanks{A preliminary version of this paper appeared at the
		15th Scandinavian Symposium and Workshops on Algorithm Theory (SWAT 2016). Author emails are \url{{kolman,hansraj}@kam.mff.cuni.cz} and \url{koutecky@iuuk.mff.cuni.cz}}.}
\received{2019-06-18}
\revised{2020-07-01}
\accepted{2020-09-15}

\author{Petr Kolman\affiliationmark{1} \and
Martin Kouteck\' y\affiliationmark{2}\thanks{The research of M. Koutecký was partially supported by Charles University project UNCE/SCI/004, by the project 19-27871X of GA ČR, and by Israel Science Foundation grant 308/18.} \and
Hans Raj Tiwary\affiliationmark{1}}
\affiliation{
	Dept. of Applied Mathematics, Faculty of Mathematics and Physics, Charles University, Prague, Czech Republic \\
	Computer Science Institute, Faculty of Mathematics and Physics, Charles University, Prague, Czech Republic}



\maketitle

\begin{abstract}
We consider the convex hull $P_{\varphi}(G)$ of
all satisfying assignments of a given MSO formula $\varphi$ on a given
graph $G$. We show that there exists an extended formulation of the polytope $P_{\varphi}(G)$ that can be described by $f(|\varphi|,\tau)\cdot n$ inequalities, where $n$ is the number of vertices in $G$,
$\tau$ is the treewidth of $G$ and $f$ is a computable function depending only on $\varphi$ and $\tau.$

In other words, we prove that the extension complexity of $P_{\varphi}(G)$ is
linear in the size of the graph $G$, with a constant depending on the treewidth
of $G$ and the formula $\varphi$.  This provides a very general yet very simple
meta-theorem about the extension complexity of polytopes related to a wide
class of problems and graphs.
As a corollary of our main result, we obtain an analogous result 
on the wider class of graphs of bounded cliquewidth.

Furthermore, we study our main geometric tool which we term the glued product of polytopes.
While the glued product of polytopes has been known since the '90s, we are the first to show that it preserves decomposability and boundedness of treewidth of the constraint matrix.
This implies that our extension of $P_\varphi(G)$ is decomposable and has a constraint matrix of bounded treewidth; so far only few classes of polytopes are known to be decomposable.
These properties make our extension useful in the construction of algorithms.
\end{abstract}

\keywords{Extension Complexity, FPT, Courcelle's Theorem, MSO Logic}


\section{Introduction}

In the '70s and '80s, it was repeatedly observed that various NP-hard
problems are solvable in polynomial time on graphs resembling trees. 
The graph property of {\em resembling a tree} was eventually formalized
as having bounded treewidth, and in the beginning of the '90s,
the class of problems efficiently solvable on graphs of bounded
treewidth was shown to contain the class of problems definable by
the Monadic Second Order Logic (MSO) (Courcelle~\cite{Courcelle:90}, 
Arnborg et al. \cite{ALS:91}, Courcelle and Mosbah \cite{CM:93}).
Using similar techniques, analogous results for weaker
logics were then proven for wider graph classes such as graphs of
bounded cliquewidth and rankwidth \cite{CMR:98}. 
Results of this kind
are usually referred to as \textit{Courcelle's theorem} for a specific
class of structures.

In this paper we study the class of problems definable by the MSO logic from
the perspective of extension complexity.
While small extended formulations are known for various special classes of 
polytopes, we are not aware of any other result in the theory of extended 
formulations that works on a wide class of polytopes the way Courcelle's 
theorem works for a wide class of problems and graphs.



\subsection{Our Contribution.}
Our contribution is two-fold.
First, we prove that satisfying assignments of an MSO formula $\varphi$ on a graph of
bounded treewidth can be {\em expressed} by a ``small'' linear program.
More
precisely, there exists a computable function $f$ such that the convex hull, which we denote
$P_{\varphi}(G)$, of satisfying assignments of $\varphi$ on a graph $G$ on $n$
vertices with treewidth $\tau$ can be obtained as the projection of a polytope described by $f(|\varphi|, \tau)\cdot n$
linear inequalities.
All our results can be extended to general finite structures where the
restriction on treewidth applies to the treewidth of their Gaifman graph~\cite{Libkin:04}.

Our proof essentially works by ``merging the common wisdom'' from the areas of 
extended formulations and fixed parameter tractability. It is known that 
dynamic programming can usually be turned into a compact extended formulation 
\cite{Kaibel11,Martin:1990}, and that Courcelle's theorem can be seen as an
instance of dynamic 
programming \cite{Kreutzer:08}, and therefore it should be
expected that the polytope of 
satisfying assignments of an MSO formula of a bounded treewidth graph be small.

However, there are a few roadblocks in trying to merge these two folklore 
wisdoms. For one, while Courcelle's theorem being an instance of dynamic 
programming in some sense may be obvious to an FPT theorist, it is far from 
clear to anyone else what that sentence may even mean. On the other hand, 
being able to turn a dynamic program into a compact polytope may be a 
theoretical possibility for an expert on extended formulations, but it is by 
no means an easy statement for an outsider to comprehend. What complicates 
the matters even further is that the result of Martin et~al.~\cite{Martin:1990} 
is not a result that can be used in a black box fashion. That is, a certain 
condition must be satisfied to get a compact extended formulation out of a 
dynamic program. This is far from a trivial task, especially for a theorem 
like Courcelle's theorem.

Our second contribution regards the main geometric tool used to prove the first result, which we term the glued product of polytopes.
This tool has been known since the work of Margot~\cite{Margot_thesis}, but we additionally show that in a particular special case it behaves well with respect to the extension complexity of polytopes, their decomposability, and treewidth.
Only few classes of polytopes have been known to be decomposable so far, as we shall discuss later.
This better understanding of the glued product then allows us to prove more about the structure of $P_\varphi(G)$, which has been useful in the construction of parameterized algorithms~\cite{KnopKMT:17,GajarskyHKO:17}, as we discuss next.

It should be noted that, similarly as in the case of Courcelle's theorem, the dependency on $\tau$ and $|\varphi|$ in the size of our extended formulation can be as bad as an exponential tower~\cite{FrickGrohe}.
However, in the case of Courcelle's theorem, it was shown that for many important graph problems (such as \textsc{Vertex Cover}, \textsc{Dominating Set} or \textsc{Coloring}) this blow-up does not occur, mainly because the number of quantifier alternations is small~\cite{LRRS:14}; using the same arguments, 
analogous observations can be made in our case as well.

\subsection{Applications}
Our results have already found at least two applications.
Knop et al.~\cite{KnopKMT:17} study the complexity of strenghtenings of the MSO logic on graphs of bounded treewidth.
They build upon our result that $P_\varphi(G)$ not only has a compact extended formulation, but that this formulation is additionally defined by a matrix with small treewidth.
This allows them to show new analogues of Courcelle's theorem for stronger logics elegantly and without explicitely dealing with the often cumbersome details.
Moreover, Koutecký shows in his PhD thesis~\cite{Koutecky:2017} constructions of compact formulations of polytopes corresponding to problems definable in these strenghtenings of MSO. 
Gajarský et al.~\cite{GajarskyHKO:17} study a broad optimization problem generalizing MSO-partitioning, which asks to partition a graph $G$ into $r$ sets satisfying a given MSO formula.
They build upon our result that $P_\varphi(G)$ has a \emph{decomposable} compact formulation, and, again, that it is definable by a matrix with small treewidth.
This allows efficient solution of a wide class of problems related to vulnerability~\cite{AssadiENYZ:14,OmranSZ:2013}, congestion~\cite{Rosenthal:1973} and others; we are not aware of any similar result.

\subsection{Organization}

The rest of the article is organized as follows. In Section \ref{sec:related_work} 
we review some previous work related to Courcelle's theorem and 
extended formulations. In Section \ref{sec:prelim} we describe the relevant 
notions related to polytopes, extended formulations, graphs, treewidth
and MSO logic. In Section~\ref{sec:glued_product} we study
several properties of the glued product of polytopes.
In Section \ref{sec:xc_mso} we prove the 
existence of compact extended formulations for $P_\varphi(G)$ parameterized by 
the length of the given MSO formula and the treewidth of the given graph. In 
Section \ref{sec:ef_construction} we describe how to efficiently construct 
such a polytope given a tree decomposition of a graph and apply our
findings from Section~\ref{sec:glued_product}. 
Finally, in Section~\ref{sec:extensions}, we prove additional properties of $P_\varphi(G)$, show applicability of our proof to graphs of bounded cliquewidth, and obtain an optimization version of Courcelle's theorem in a particularly simple way.

\section{Related Work}\label{sec:related_work}

\subsection{MSO Logic vs. Treewidth}

Because of the wide relevance of the treewidth parameter in many
areas (cf. survey by Bodlaender~\cite{Bodlaender:06}) and the large
expressivity of the MSO and its extensions (cf. the survey of
Langer et al.~\cite{LRRS:14}), considerable attention was given to
Courcelle's theorem by theorists from various fields, reinterpreting
it into their own setting. These reinterpretations helped uncover several
interesting connections.

The classical way of proving Courcelle's theorem is constructing a tree
automaton $A$ in time only dependent on $\varphi$ and the treewidth $\tau$,
such that $A$ accepts a tree decomposition of a graph of treewidth $\tau$ if
and only if the corresponding graph satisfies $\varphi$; this is the automata
theory perspective~\cite{Courcelle:90}. Another perspective comes
from finite model theory
where one can prove that a certain equivalence on the set of graphs of
treewidth at most $\tau$ has only finitely many (depending on $\varphi$ and
$\tau$) equivalence classes and that it {\em behaves well} \cite{GPW:07}.
Another approach proves that a quite different equivalence on so-called
extended model checking games has finitely many equivalence
classes~\cite{KLR:11} as well; this is the game-theoretic perspective. It can
be observed that the finiteness in either perspective stems from the same
roots.

Another related result is an {\em expressivity} result: Gottlob et
al.~\cite{GPW:07} prove that on bounded treewidth graphs, a certain subset of
the database query language Datalog has the same expressive power as the MSO.
This provides an interesting connection between automata theory and 
database theory.

We note that several implementations of Courcelle's theorem have been developed in the recent years which are usable on quite large instances and even beat ILP solvers on specific problems such as \textsc{Connected Dominating Set}~\cite{BannachB19,LangerRRS12,LRRS:14}.

\subsection{Extended Formulations}

Sellmann, Mercier, and Leventhal \cite{SML:07} claimed to show compact extended 
formulation for binary Constraint Satisfaction Problems (CSP) for graphs of bounded treewidth, but their proof is not 
correct~\cite{Sellmann:08}.
The first two authors of this paper gave extended 
formulations for CSP that has polynomial size 
for instances whose constraint graph has bounded treewidth~\cite{KK:15} using 
a different technique. Bienstock and Munoz~\cite{BM:15} prove similar results
for 
the approximate and exact version of the problem. In the exact case, Bienstock 
and Munoz's bounds are slightly worse than those of Kolman and Kouteck\'y~\cite{KK:15}. 
It is worth noting that the satisfiability of CSPs of constant domain 
size can be expressed in MSO logic.
Laurent~\cite{Laurent:09} provides extended formulations for the
independent set and max cut polytopes of size $\Oh(2^{\tau}n)$ for
$n$-vertex graphs of treewidth $\tau$
and, independently, Buchanan and Butenko \cite{BB:14} provide
an extended formulation for the independent set polytope of the same
size.
Gajarský et al.~\cite{GajarskyHT18} have shown that the stable set polytope has a compact extended formulation for graphs of bounded expansion, and that it has no extended formulation of size $f(k) \cdot \textrm{poly}(n)$ for any function $f$ where $k$ is the size of the stable set.
Faenza et al.~\cite{FaenzaMP2018} study the limits of using treewidth for obtaining extended formulations.

A lot of recent work on extended formulations has focussed on establishing 
lower bounds in various settings: exact, approximate, linear vs. 
semidefinite, etc. (See for example \cite{FMPTW15,AT2013,BraunFPS15,LeeRS15}). A wide variety of tools have been developed 
and used for these results including connections to nonnegative matrix 
factorizations \cite{Yannakakis91}, communication complexity~\cite{FaenzaFGT15}, information theory 
\cite{BraunJLP17}, and quantum communication \cite{FMPTW15} among others.

For proving upper bounds on extended formulations, several authors have 
proposed various tools as well. Kaibel and Loos \cite{KaibelL10} describe a setting of 
branched polyhedral systems which was later used by Kaibel and Pashkovich 
\cite{KaibelP11} to provide a way to construct polytopes using reflection relations. 

A specific composition rule, which we term glued product (cf.
Section~\ref{sec:glued_product}), was studied by Margot in his PhD thesis \cite{Margot_thesis}. Margot
showed that a property called the projected face property suffices to glue two
polytopes efficiently. Conforti and Pashkovich~\cite{CP12} describe and strengthen
Margot's result to make the projected face property to be a necessary and
sufficient condition to describe the glued product in a particularly efficient
way. 

Martin et al.~\cite{Martin:1990} have shown that under certain conditions, an
efficient dynamic programming based algorithm can be turned into a compact
extended formulation. Kaibel \cite{Kaibel11} summarizes this and various other methods.
Other results showing some kind of general way of constructing a compact extended formulations are due to Fiorini et al.~\cite{FioriniHW2017}, who define a hierarchy of increasingly accurate polytopes approximating some target polytope, and Aprile and Faenza~\cite{AprileF2019}, who show an output-efficient way to construct an extended formulation from a communication protocol.

\section{Preliminaries} \label{sec:prelim}

\subsection{Polytopes, Extended Formulations and Extension Complexity}\label{subsec:polytopes}

For background on polytopes we refer the 
reader to Gr\"unbaum \cite{gruenbaum} and Ziegler \cite{ziegler}. 
To simplify reading of the paper for audience not working 
often in the area of polyhedral combinatorics, we provide here a brief
glossary of common polyhedral notions that are used in this article.

 A {\em hyperplane} in $\mathbb{R}^n$ is a closed convex set of the form
$\{x|a^\intercal x = b\}$ where $a \in \RR^n, b\in\RR$.
 A {\em halfspace} in $\mathbb{R}^n$ is a closed convex set of the form
$\{x|a^\intercal x\leqslant b\}$ where $a \in \RR^n, 
b\in\RR$. The inequality $a^\intercal x\leqslant b$ is said to define
the corresponding halfspace.  
A {\em polytope} $P\subseteq\RR^n$ is a bounded
subset defined by intersection of finite number of halfspaces. 
A polytope $P$ is a \emph{$0/1$ polytope} if $P$ is the convex hull of a finite set of $0/1$-vectors.
A result of
Minkowsky-Weyl states that equivalently, every polytope is the convex hull of a
finite number of points.
Let $h$ be a halfspace defined by an
inequality $a^\intercal x\leqslant b$; the inequality is said to be {\em valid} for a
polytope $P$ if $P=P\cap h.$
Let $a^\intercal x \leqslant b$ be a  valid inequality for polytope $P$; 
then, $P\cap \{x|a^\intercal x=b\}$ is said to be
a {\em face} of $P$.  
For a vector $x \in \RR^{n}$, we use $x[i]$ to denote its $i$-th coordinate.

Note that, taking $a$ to be the zero vector and $b=0$ results in the face being
$P$ itself. Also, taking $a$ to be the zero vector and $b=1$ results in the
empty set. These two faces are often called the trivial faces and they are polytopes
``living in'' dimensions $n$ and $-1$, respectively. Every face -- that is not trivial -- is itself a polytope of dimension  $d$ where $0\leqslant d\leqslant
n-1$.

It is not uncommon to refer to three separate (but related) objects as a face:
the actual face as defined above, the valid inequality defining it, and the
equation corresponding to the valid inequality. While this is clearly a misuse
of notation, the context usually makes it clear as to exactly which object is
being referred to.

The zero dimensional faces of a polytope
are called its {\em vertices}, and the $(n-1)$-dimensional faces are called its
{\em facets}. 

Let $P$ be a polytope in $\RR^d$. A polytope $Q$ in $\RR^{d+r}$ is called an 
\textit{extended formulation} or an \textit{extension} of $P$ if $P$ is a
projection of $Q$ 
onto the first $d$ coordinates.
Note that for any linear map $\pi:\RR^{d+r}\to
\RR^d$ such that $P=\pi(Q)$, there exists a polytope $Q'$ with the same number of facets as $Q$, such that $P$ is 
obtained by dropping all but the first $d$ coordinates on $Q'$.

The \textit{size} of a polytope is defined to be the number of its
facet-defining 
inequalities. Finally, the \emph{extension complexity} of a polytope $P$,
denoted by $\xc(P)$, is the size of its smallest extended formulation. We 
refer the readers to the surveys 
\cite{CCZ:13,VanderbeckWolsey2010,Kaibel11,Wolsey11} 
for details and background of the subject and
we only state two basic propositions about extended formulations here.

\begin{proposition}\label{prop:xc_vertices} 
	Let $P$ be a polytope with a vertex set $V=\left\{v_1,\ldots,v_n\right\}$. 
Then $\xc(P)\leqslant n.$
\end{proposition}

\begin{proof}
Let $P=\conv\left(\left\{v_1,\ldots,v_n\right\}\right)$ be a 
polytope. Then, $P$ is the projection of 
$$Q=\left\{\left(x,\lambda\right) 
\left| x=\sum\limits_{i=1}^n \lambda_i v_i; \sum\limits_{i=1}^n \lambda_i = 1
; \lambda_i \geqslant 0 \mbox{ for } i\in\{1,\ldots,n\} \right.\right\}.$$ 
It is clear that $Q$ has at most $n$ facets and therefore $\xc(P)\leqslant n$.
\end{proof}

\begin{proposition}\label{prop:xc_slice}
	Let $P$ be a polytope obtained by intersecting a set $H$ of hyperplanes with  a
	polytope~$Q$. Then $\xc(P)\leqslant \xc(Q)$.
\end{proposition}
\begin{proof}
Note that any extended formulation of $Q$, when intersected with $H$, 
gives an extended formulation of $P$. Intersecting a polytope with 
hyperplanes does not increase the number of facet-defining inequalities (and 
only possibly reduces it). 
\end{proof}

\subsection{Graphs and Treewidth}

For notions related to the treewidth of a graph and nice tree
decomposition, in most cases we stick to the standard terminology as given in
the book by Kloks~\cite{Kloks:94}; the only deviation is in the leaf nodes of
the nice tree decomposition where we assume that the bags are empty.
For a vertex $v\in V$ of a graph $G=(V,E)$, we denote by $\delta_G(v)$ the
set of neighbors of $v$ in $G$, that is, $\delta_G(v)=\left\{u\in V\ | \ \{u,v\}\in E\right\}$. If the graph $G$ is clear from the context, we omit the subscript and simply write $\delta(v)$.

A \emph{tree decomposition} of a graph $G=(V,E)$ 
is a pair $(T, B)$, where $T$ is a rooted tree and $B$ is a mapping
$B: V(T) \rightarrow 2^V$ satisfying
\begin{itemize}
	\setlength\itemsep{0em}
	\item for any $uv \in E$, there exists $a \in V(T)$ such that
	$u, v \in B(a)$,
	\item if $v \in B(a)$ and $v \in B(b)$, then $v \in B(c)$ for all
	$c$ on the path from $a$ to $b$ in $T$.
\end{itemize}
We use the convention that the vertices of the tree are called nodes and the sets
$B(a)$ are called \emph{bags}.
Occasionally, we will view the mapping $B$ as the set $B = \{B(u) \mid u \in V(T)\}$.

The {\em treewidth $tw((T, B))$ of a tree decomposition} $(T, B)$ is
the size of the largest bag of $(T, B)$ minus one.
The {\em treewidth $tw(G)$ of a graph} $G$ is the
minimum treewidth over all possible tree decompositions of $G$.


A \emph{nice tree decomposition} 
is a tree decomposition with one
special node $r$ called the \emph{root} in which each node is
one of the following types:
\begin{itemize}
	\setlength\itemsep{0em}
	\item \emph{Leaf node}: a leaf $a$ of $T$ with $B(a) = \emptyset$.
	\item \emph{Introduce node}: an internal node $a$ of $T$ with one
	child $b$ for which $B(a) = B(b) \cup \{v\}$ for some $v \not\in
	B(a)$.
	\item \emph{Forget node}: an internal node $a$ of $T$ with one child
	$b$ for which $B(a) = B(b) \setminus \{v\}$ for some $v \in B(b)$.
	\item \emph{Join node}: an internal node $a$ with two children $b$
	and $c$ with $B(a) = B(b) = B(c)$.
\end{itemize}
For a vertex $v\in V$, we denote by $top(v)$ the topmost
node of the nice tree decomposition $(T, B)$ that contains $v$ in its bag. 
For any graph $G$ of treewith $\tau$ on $n$ vertices, a nice tree decomposition of $G$ of width $\tau$ 
with at most $8n$ nodes can be computed
in time $f(\tau) \cdot n$, for some computable function $f$~\cite{Bodlaender:93,Kloks:94}.

Given a graph $G=(V,E)$ and a subset of vertices $\{v_1, \dots, v_d\}\subseteq V$,
we denote by $G[v_1, \dots, v_d]$ the subgraph of $G$ induced by the
vertices $v_1, \dots, v_d$. 
Given a tree decomposition $(T, B)$ and a node $a \in V(T)$, 
we denote by $T_a$ the subtree of $T$ rooted in
$a$, and by $G_a$ the subgraph of $G$ induced by all vertices in
bags of $T_a$, that is, $G_a = G[\bigcup_{b \in V(T_a)} B(b)]$.
Throughout this paper we assume that for every graph, its vertex set is 
a subset of $\NN$. We define the following operator $\eta$: 
for any set $U = \{v_1, v_2, \dots, v_l\} \subseteq \NN$, 
$\eta(U) = (v_{i_1}, v_{i_2}, \dots, v_{i_l})$ such that $v_{i_1} < v_{i_2} \dots < v_{i_l}$; in other words, $\eta$ takes the set $U$ to the ordered tuple of its elements.

\vspace{2ex}
\begin{figure}[h]
\centering
    \includegraphics[scale=0.5]{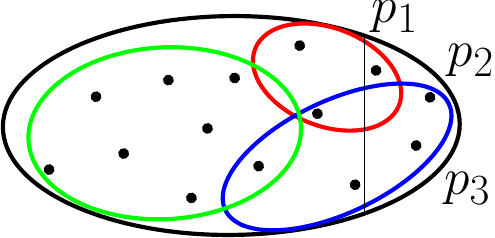}
    \caption{A $[3]$-labelled $3$-boundaried graph with $\vec{p}=(p_1,p_2,p_3)$.}
\end{figure}
\medskip

For an integer $m\geq 0$, an \textit{$[m]$-labelled graph} is a pair $(G,
\vec{V})$ where $G=(V,E)$ is a graph and $\vec{V} = (V_1, \dots, V_m)$ is an
$m$-tuple of subsets of vertices of $G$ called an {\em $m$-labelling of $G$}.
For a subset of vertices $W \subseteq V$, we denote by $\vec{V}[W]$ the restriction of $\vec{V}$ to $W$, i.e., $\vec{V}[W] = (V_1 \cap W, \dots, V_m \cap W)$.
For integers $m\geq 0$  and $\tau\geq 0$, an \textit{$[m]$-labelled
$\tau$-boundaried graph} is a triple $(G, \vec{V}, \vec{p})$ where $(G,
\vec{V})$ is an $[m]$-labelled graph and $\vec{p} = (p_1, \dots, p_\tau)$ is a
$\tau$-tuple of vertices of $G$ called a \textit{boundary} of $G$.  If the
tuples $\vec{V}$ and $\vec{p}$ are clear from the context or if their content
is not important, we simply denote an $[m]$-labelled $\tau$-boundaried graph 
by $G^{[m],\tau}$. For a tuple $\vec{p} = (p_1, \dots, p_\tau)$, we denote
by $p$ the corresponding set, that is, $p = \{p_1, \dots, p_\tau\}$.
We say that $\vec{p}'$ is a subtuple of $\vec{p}$ if $p' \subseteq p$.

Two $[m]$-labelled $\tau$-boundaried
{\em graphs} $(G_1, \vec{V}, \vec{p})$ and $(G_2, \vec{U}, \vec{q})$ are
\textit{compatible} if the function $h: \vec{p} \rightarrow \vec{q}$,
defined by $h(p_i) = q_i$ for each $i$, is an isomorphism of the
induced subgraphs $G_1[p_1, \dots, p_\tau]$ and $G_2[q_1, \dots,
  q_\tau]$, and if for each $i$ and $j$, $p_i \in V_j \Leftrightarrow
q_i \in U_j$.

\vspace{2ex}
\begin{figure}[h]
\centering
  \includegraphics[scale=0.5]{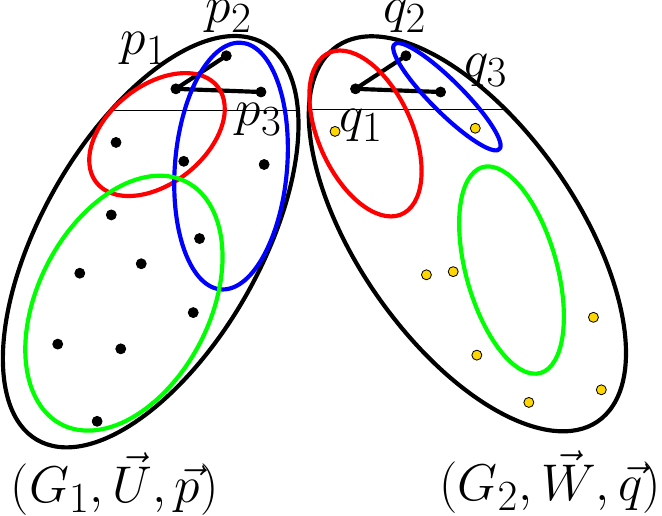}
  \caption{Compatibility of two $[3]$-labelled $3$-boundaried graphs.}
\end{figure}

Given two compatible $[m]$-labelled $\tau$-boundaried graphs $G_1^{[m],\tau} =
(G_1, \vec{U}, \vec{p})$ and $G_2^{[m],\tau} = (G_2, \vec{W}, \vec{q})$, the
\textit{join} of $G_1^{[m],\tau}$ and $G_2^{[m],\tau}$, denoted by
$G_1^{[m],\tau} \oplus G_2^{[m],\tau}$, is the $[m]$-labelled $\tau$-boundaried
graph $G^{[m],\tau} = (G, \vec{V}, \vec{p})$ where
\begin{itemize}
\setlength\itemsep{0em}
\item $G$ is the graph obtained by taking the disjoint union of $G_1$ and
$G_2$, and for each $i$, identifying the vertex $p_i$ with the vertex $q_i$ 
and keeping the label $p_i$ for it;
\item $\vec{V} = (V_1, \dots, V_m)$ with $V_j = U_j \cup W_j$ 
	and every $q_i$ replaced by $p_i$, for each $j$ and $i$;
\item $\vec{p} = (p_1, \dots, p_\tau)$ with $p_i$ being the node in $V(G)$ 
obtained by the identification of $p_i\in V(G_1)$ and $q_i\in V(G_2)$, for each
$i$.
\end{itemize} 
\medskip
\begin{figure}
\centering
\includegraphics[scale=0.5]{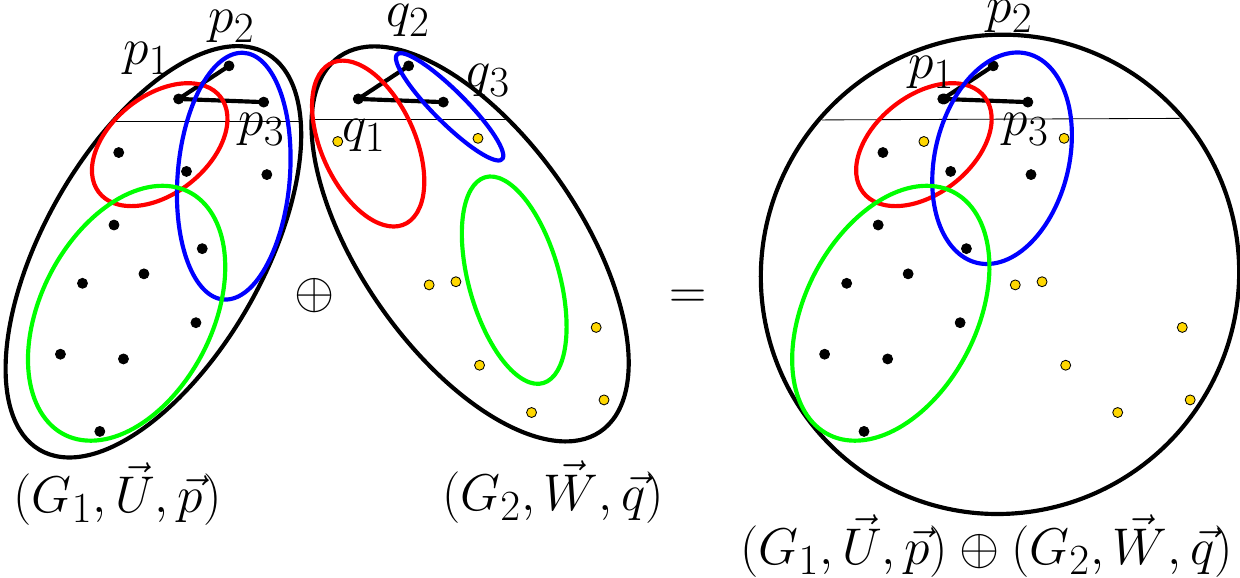}
\caption{The \textit{join} of two $[m]$-labelled $\tau$-boundaried graphs.}
\end{figure}
Because of the choice of referring to the boundary vertices by their names in
$G_1^{[m],\tau}$, it does not always hold that
$G_1^{[m],\tau} \oplus G_2^{[m],\tau} = G_2^{[m],\tau} \oplus G_1^{[m],\tau}$;
however, the two structures are isomorphic and equivalent for our purposes (see
below).

\subsection{Monadic Second Order Logic and Types of Graphs} \label{subsec:MSO_types}

In most cases, we stick to standard notation as given by
Libkin~\cite{Libkin:04}.
A \textit{vocabulary} $\sigma$ is a finite collection of \textit{constant}
symbols $c_1, c_2, \ldots$ and \textit{relation} symbols $P_1, P_2,
\ldots$. Each relation symbol $P_i$ has an associated arity $r_i$. A
\textit{$\sigma$-structure} is a tuple $\AA = (A, \{c_i^{\AA}\}, \{P_i^{\AA}\})$
that consists
of a universe $A$ together with an interpretation of the constant
and relation symbols: each constant symbol $c_i$ from $\sigma$ is associated with 
an element $c_i^{\AA} \in A$ and each relation symbol $P_i$ from $\sigma$ 
is associated with an $r_i$-ary relation $P_i^{\AA} \subseteq A^{r_i}$;
note that, with the exception of $\tau$-boundaried graphs, the set of constants of structures studied here is always empty.

To give an example, a graph $G=(V,E)$ can be viewed as a $\sigma_{1}$-structure 
$(V,\emptyset,\{E\})$ where $E$ is a symmetric binary relation on $V\times V$
and the vocabulary $\sigma_{1}$ contains a single relation symbol.
Alternatively, for another
vocabulary $\sigma_2$ containing three relation symbols,
one of arity two and two of arity one, one can view a
graph $G=(V,E)$ also as a $\sigma_{2}$-structure $I(G)=(V_I,\emptyset,\{E_I,L_V,L_E\})$, 
with $V_I=V\cup E$, $E_I=\{\{v,e\}\ | \ v\in e, e\in E\}$, $L_V=V$
and $L_E=E$; we will call $I(G)$ the {\em incidence graph} of $G$.
In our approach we will make use of the well known fact that the
treewidths of $G$ and $I(G)$, viewed as a $\sigma_1$- and $\sigma_2$-structures
as explained above, differ by at most one~\cite{KV:98}.
(As mentioned earlier, we define the treewidth of a structure $\AA$ as the treewidth of its \emph{Gaifman graph $G(\AA) = (A,E)$} with $E = \{\{u,v\} \mid u,v \in p \in P_i^{\AA} \}$, i.e., two vertices are connected by an edge if the corresponding elements from the universe appear together in some relation.)

The main subject of this paper are formulae for graphs in monadic second-order logic (MSO) which is an extension of first-order logic that allows quantification over subsets of elements of the universe.
We denote by MSO$_1$ the MSO logic over the signature $\sigma_1$ and
by MSO$_2$ the MSO logic over the signature $\sigma_2$.
(See the Appendix for a full formal definitions of formulae in MSO$_1$ and MSO$_2$.)
For example, the $3$-colorability property can be expressed in MSO$_1$ as follows:
\begin{eqnarray*}
	\varphi_{\textrm{3-col}} \equiv \exists X_1,X_2,X_3 &&\left[\,
	\forall x \, (x\in X_1\vee x\in X_2\vee x\in X_3) \wedge \right.
	\\ &&\bigwedge\nolimits_{i=1,2,3} \left. \!\!
	\forall x,y
	\left(x\not\in X_i\vee y\not\in X_i\vee \neg\mathrm{E}(x,y)\right)
	\,\right] \enspace .
\end{eqnarray*}
When a structure $\AA$ satisfies a formula $\varphi$ we write $\AA \models \varphi$. 
For example, a graph $G$ is $3$-colorable if and only if $G \models \varphi_{\textrm{3-col}}$.

An important kind of structures that are necessary in the proofs in this paper are the
$[m]$-labelled $\tau$-boundaried graphs. An $[m]$-labelled 
$\tau$-boundaried graph $G=(V,E)$ with boundary $p_1,\ldots,p_{\tau}$
labelled with $V_1,\ldots,V_m$ is viewed as a structure
$(V_I,\{p_1,\ldots,p_{\tau}\},\{E_I,L_V,L_E,$ $V_1,\ldots,V_m\})$;
for notational simplicity, we stick to the notation $G^{[m],\tau}$
or $(G,\vec V, \vec p)$.
The corresponding vocabulary is denoted by $\sigma_{m,\tau}$.  

A variable $X$ is {\em free} in $\varphi$ if it does not appear in any
quantification in $\varphi$. 
If $\vec{X}$ is the tuple of all free variables in $\varphi$,
we write $\varphi(\vec{X})$.
A variable $X$ is \textit{bound} in $\varphi$ if it is not free.
To simplify the presentation, without loss of generality (cf.~\cite{KLLL:15}) we assume 
that the free variables of the input formulae are only set variables
(no free element variables). 

By $qr(\varphi)$ we denote the \textit{quantifier rank} of
$\varphi$ which is the maximum depth of nesting of quantifiers in $\varphi$.
Two structures $\AA_1, \AA_2$ over the same vocabulary $\sigma$ are \emph{$\MSOk$-elementarily equivalent} if they satisfy the same $\MSO$ formulae over $\sigma$ with quantifier depth at most $k$; this is denoted $\AA_1 \MSOkeq \AA_2$. 	
The main tool in the model theoretic approach to Courcelle's theorem,
that will also play a crucial role in our approach, can be stated as
the following theorem which follows from \cite[Proposition~7.5 and Theorem~7.7]{Libkin:04}.
\begin{theorem}[\cite{Libkin:04}] 
\label{thm:finite_index} 
For any fixed vocabulary $\sigma$ and $k \in \NN$,
the equivalence relation $\equivMSOkt$ has a finite number of equivalence classes.
  \end{theorem}
A specialization to $[m]$-labelled $\tau$-boundaried graphs is the following.
We denote by $\MSO[k,\tau,m]$ the set of all $\MSO$ formulae $\varphi$ 
over the vocabulary $\sigma_{m, \tau}$ with $qr(\varphi) \leq k$.
Two $[m]$-labelled $\tau$-boundaried graphs $G_1^{[m], \tau}$ and
$G_2^{[m], \tau}$ are
$\MSOk$-elementarily equivalent if they satisfy the same
$\MSO[k,\tau,m]$ formulae, and we have
\begin{corollary}
	\label{cor:finite_index} 
	For any fixed $\tau, k,
	m \in \NN$, the equivalence relation $\equivMSOkt$ over $[m]$-labelled $\tau$-boundaried graphs has a finite
	number of equivalence classes.
\end{corollary}
Let us continue the example of $\varphi_{\textrm{3-col}}$.
We have $m=3$ free variables and quantifier rank $k=qr(\varphi_{\textrm{3-col}}) = 5$.
The rough intuition of what the equivalence classes of $\MSOkeq$ look like is the following: we have an equivalence class for each possible labeling of the boundary.
This can be simplified with the knowledge that in $\varphi_{\textrm{3-col}}$, the satisfying assignments are always actual colorings (each vertex belongs to exactly one of the sets $V_1, V_2, V_3$), so the equivalence classes would essentially be graphs with identical colorings of the boundary (i.e., all graphs with a given coloring of the boundary fall into the same equivalence class).  As we will shortly see, the handy property is that, given two $3$-colored $\tau$-boundaried graphs $G_1$ and $G_2$ whose colorings agree on their boundaries, we are guaranteed that their join is also $3$-colorable.

Let us denote by $\CC$ the equivalence classes of the relation $\equivMSOkt$, fixing an ordering such that $\alpha_1$ is the
class containing the empty graph. Note that the size of $\CC$ depends only on
$k$, $m$ and $\tau$, that is, $|\CC| = f(k, m, \tau)$ for some computable
function $f$.
Let us denote by $\CCtw$ the equivalence classes of $\equivMSOkt$ on $[m]$-labelled $\tau$-boundaried graphs of treewidth at most $\tau$; notice that $\CCtw \subseteq \CC$.
For a given MSO formula $\varphi$ with $m$ free variables, we define an {\em indicator
function} $\rho_\varphi:\{1,\ldots,|\CC|\} \to \{0,1\}$ as follows: for every $i$,
if there exists a graph $G^{[m],\tau} \in \alpha_i$ such that 
$G^{[m],\tau} \models \varphi$, we set $\rho_\varphi(i)=1$, and 
we set $\rho_\varphi(i)=0$ otherwise; note that if there exists a graph
$G^{[m],\tau} \in \alpha_i$ such that $G^{[m],\tau} \models \varphi$,
then $G'^{[m],\tau} \models \varphi$ for every $G'^{[m],\tau} \in \alpha_i$.

For every $[m]$-labelled $\tau$-boundaried graph $G^{[m],\tau}$, 
its \textit{type}, with respect
to the relation $\equivMSOkt$, is the class to which $G^{[m],\tau}$
belongs.
We say that {\em types} $\alpha_i$ and $\alpha_j$ are
\textit{compatible} if there exist two $[m]$-labelled $\tau$-boundaried 
graphs of types $\alpha_i$ and
$\alpha_j$ that are compatible; note that this is well defined as all 
$[m]$-labelled $\tau$-boundaried graphs of
a given type are compatible. For every $i\geq 1$, we will encode the type
$\alpha_i$ naturally as a binary vector $\{0,1\}^{|\CC|}$ with exactly one $1$,
namely with $1$ on the position $i$.

An important property of the types and the join operation is that the type
of a join of two $[m]$-labelled $\tau$-boundaried graphs depends on their types
only.
\begin{lemma}[{\cite[Lemma~7.11]{Libkin:04} and \cite[Lemma~3.5]{GPW:07}}]
\label{lem:only_types_matter}
  Let $G_a^{[m],\tau}$, $G_{a'}^{[m],\tau}$, $G_{b}^{[m],\tau}$ and
  $G_{b'}^{[m],\tau}$ be $[m]$-labelled $\tau$-boundaried
  graphs such that $G_a^{[m],\tau} \equivMSOkt G_{a'}^{[m],\tau}$ 
  and $G_{b}^{[m],\tau} \equivMSOkt G_{b'}^{[m],\tau}$. 
  Then $(G_a^{[m],\tau} \oplus G_b^{[m],\tau}) \equivMSOkt
  (G_{a'}^{[m],\tau} \oplus G_{b'}^{[m],\tau})$. 
\end{lemma}
The importance of the lemma rests in the fact that for determination of
the type of a join of two $[m]$-labelled $\tau$-boundaried graphs, 
it suffices to know only a {\em small} amount of information about the two
graphs,
namely their types.
The following two lemmas deal in a similar way with the type of a graph
in other situations. 
\begin{lemma}[{\cite[implicit]{GPW:07}}] \label{lem:introduce_matter}
  Let $(G_a,\vec{X},\vec{p})$, $(G_{b},\vec{Y},\vec{q})$ be
  $[m]$-labelled $\tau$-boundaried graphs and
  let $(G_{a'},\vec{X'},\vec{p'})$, $(G_{b'},\vec{Y'},\vec{q'})$ be
  $[m]$-labelled $(\tau+1)$-boundaried graphs with 
  $G_{a}=(V,E)$, $G_{a'}=(V',E')$, 
  $G_{b}=(W,F)$, $G_{b'}=(W',F')$ 
  such that 
\begin{enumerate}
\setlength\itemsep{0em}
\item $(G_a,\vec{X},\vec{p}) \equivMSOkt (G_{b},\vec{Y},\vec{q})$;
\item $V'=V\cup\{v\}$ for some $v\not\in V$, $\delta(v)\subseteq p$, $\vec{p}$ is a subtuple of $\vec{p'}$
  and $(G_{a'}[V],\vec{X'}[V],\vec{p'}[V])=(G_a,\vec{X},\vec{p})$;
\item $W'=W\cup\{w\}$ for some $w\not\in W$, $\delta(w)\subseteq q$, $\vec{q}$ is a subtuple of $\vec{q'}$
  and $(G_{b'}[W],\vec{Y'}[W],\vec{q'}[W])=(G_{b},\vec{Y},\vec{q})$; 
\item $(G_{a'},\vec{X'},\vec{p'})$ and $(G_{b'},\vec{Y'},\vec{q'})$ are
compatible.
\end{enumerate}
  Then $(G_{a'},\vec{X'},\vec{p'}) \equivMSOktplus (G_{b'},\vec{Y'},\vec{q'})$.
\end{lemma}

\begin{lemma}[{\cite[implicit]{GPW:07}}] \label{lem:forget_matter}
  Let $(G_a,\vec{X},\vec{p})$, $(G_{b},\vec{Y},\vec{q})$ be
  $[m]$-labelled $\tau$-boundaried graphs and
  let $(G_{a'},\vec{X'},\vec{p'})$, $(G_{b'},\vec{Y'},\vec{q'})$ be
  $[m]$-labelled $(\tau+1)$-boundaried graphs with 
  $G_{a}=(V,E)$, $G_{a'}=(V',E')$, 
  $G_{b}=(W,F)$, $G_{b'}=(W',F')$ 
  such that 
\begin{enumerate}
\setlength\itemsep{0em}
\item $(G_{a'},\vec{X'},\vec{p'}) \equivMSOktplus (G_{b'},\vec{Y'},\vec{q'})$;
\item $V\subseteq V'$, $|V'|=|V|+1$, $\vec{p}$ is a subtuple of $\vec{p'}$
  and $(G_{a'}[V],\vec{X'}[V],\vec{p'}[V])=(G_a,\vec{X},\vec{p})$;
\item $W\subseteq W'$, $|W'|=|W|+1$, $\vec{q}$ is a subtuple of $\vec{q'}$
  and $(G_{b'}[W],\vec{Y'}[W],\vec{q'}[W])=(G_{b},\vec{Y},\vec{q})$.
\end{enumerate}
  Then $(G_a,\vec{X},\vec{p}) \equivMSOkt (G_{b},\vec{Y},\vec{q})$.
\end{lemma}

\subsection{Feasible Types} \label{subsec:MSO_treewidth} 

Suppose that we are given an MSO$_2$ formula $\varphi$ 
with $m$ free variables and a quantifier rank at most $k$,
a graph $G$ of treewidth at most $\tau$, and a nice tree decomposition $(T, B)$ of
the graph $G$.  

For every node of $T$ we are going to define certain types and
tuples of types as \textit{feasible}.
For a node $b \in V(T)$ of any kind (leaf, introduce, forget, join)
and for $\alpha \in \CC$, we say that $\alpha$ is a \textit{feasible
type of the node $b$} if there exist $X_1,\ldots,X_m\subseteq V(G_b)$
such that $(G_b, \vec{X}, \eta(B(b)))$ is of type $\alpha$
where $\vec{X}=(X_1,\ldots,X_m)$;
we say that $\vec{X}$ \textit{realizes type $\alpha$ on the node $b$}. 
We denote the set of feasible types of the node $b$ by $\FF(b)$.

For an {\em introduce} node $b \in V(T)$ with a
child $a \in V(T)$ (assuming that $v$ is the new vertex), for
$\alpha \in \FF(a)$ and $\beta \in \FF(b)$, we
say that $(\alpha, \beta)$ is a \textit{feasible pair of types for
$b$} if there exist $\vec{X}=(X_1,\ldots,X_m)$ and
$\vec{X'}=(X'_1,\ldots,X'_m)$ 
realizing types $\alpha$
and $\beta$ on the nodes $a$ and $b$, respectively, such that for each
$i$, either $X'_i = X_i$ or $X'_i = X_i \cup \{v\}$. 
We denote the set of feasible pairs of types of the introduce node 
$b$ by $\FF_p(b)$.

For a {\em forget} node $b \in V(T)$ with a
child $a \in V(T)$ and for $\beta\in \FF(b)$ and $\alpha\in \FF(a)$,
we say $(\alpha, \beta)$ is a \textit{feasible pair of
  types for $b$} if there exists $\vec{X}$ realizing $\beta$ on $b$
and $\alpha$ on $a$.
We denote the set of feasible pairs of types of the forget node 
$b$ by $\FF_p(b)$.

For a {\em join} node $c \in V(T)$ with children $a, b
\in V(T)$ and for $\alpha\in \FF(c)$, $\gamma_1\in \FF(a)$ and 
$\gamma_2\in \FF(b)$,
we say that ($\gamma_1, \gamma_2, \alpha)$ is a
\textit{feasible triple of types for $c$} if $\gamma_1$, $\gamma_2$
and $\alpha$ are mutually compatible and there exist $\vec{X^1},
\vec{X^2}$ realizing $\gamma_1$ and $\gamma_2$ on $a$ and $b$,
respectively, such that $\vec{X} = (X_1^1 \cup X_1^2, \dots, X_m^1
\cup X_m^2)$ realizes $\alpha$ on $c$.
We denote the set of feasible triples of types of the join node $c$ by
$\FF_t(c)$.

We define an indicator function 
$\nu: \CC \times V(T) \times V(G) \times \{1, \ldots, m\} \rightarrow \{0,1\}$ such that $\nu(\beta, b, v, i) = 1$ if and only if there exists $\vec X=(X_1,\ldots,X_m)$ realizing the type $\beta$ on the node 
$b \in V(T)$ with $v \in B(b)$ and $v \in X_i$.
Additionally, we define 
$\mu: \CC \times V(G) \times \{1,\ldots,m\} 
\rightarrow \{0,1\}$ to be $\mu(\beta,v,i) = \nu(\beta, top(v), v, i)$.

\section{Glued product of Polytopes over Common Coordinates}\label{sec:glued_product}

The (cartesian) {\em product} of two polytopes $P_1$ and $P_2$ is defined as
$$\displaystyle P_1\times P_2=\conv\left(\left\{(x,y)\mid x\in P_1, y\in P_2\right\}\right).$$ 

\begin{proposition}\label{prop:xc_product}
Let $P_1,P_2$ be two polytopes. Then
$$\xc(P_1\times P_2)\leqslant \xc(P_1)+\xc(P_2)\ .$$
\end{proposition}
\begin{proof}
Let $Q_1$ and $Q_2$ be extended formulations of $P_1$ and $P_2$, respectively. 
Then, $Q_1\times Q_2$ is an extended formulation of $P_1\times P_2$.
Now assume that $Q_1=\{x~|~Ax\leqslant b\}$ and $Q_2=\{y~|~Cy\leqslant d\}$
and that these are the smallest extended formulations of $P_1$ and $P_2$, resp.
Then, 
$$Q_1\times Q_2 = \{(x,y)~|~Ax\leqslant b, Cy\leqslant d\} \ .$$
That is, we have an extended formulation of $P_1\times P_2$ of size at most 
$\xc(P_1)+\xc(P_2)$.
\end{proof}

We are going to define the glued product of polytopes, a slight
generalization of the usual product of polytopes. We study a case where 
the extension complexity of the glued product of two polytopes is upper bounded
by the sum of the extension complexities of the two polytopes and
which exhibits several other nice properties. Then we use
it in Section \ref{sec:xc_mso} to describe a small extended
formulation for $P_{\varphi}(G)$ on graphs with bounded
treewidth. 

Let $P\subseteq \mathbb{R}^{d_1+k}$ and $Q\subseteq \mathbb{R}^{d_2+k}$ be
$0/1$-polytopes defined by $m_1$ and $m_2$ inequalities and with vertex sets
$\vertexset(P)$ and $\vertexset(Q)$, respectively.
Let $I_P\subseteq\{1,\ldots, d_1+k\}$ be a subset of coordinates of size $k$,
$I_Q\subseteq\{1,\ldots, d_2+k\}$ be a subset of coordinates of size $k$,
and let $I_P'=\{1,\ldots, d_1+k\}\setminus I_P$.
For a vector $x$, and a subset $I$ of coordinates, we denote by $x|_{I}$ the subvector of $x$
specified by the coordinates $I$.
The {\em glued product} of
$P$ and $Q$, (glued) with respect to the $k$ coordinates $I_P$ and $I_Q$,
denoted by
$P\times_k Q$, is defined as $$\displaystyle P\times_k
Q=\conv\left(\left\{(x|_{I_P'},y)\in\mathbb{R}^{d_1+d_2+k}\mid x\in \vertexset(P),
y\in \vertexset(Q), x|_{I_P}=y|_{I_Q}\right\}\right).$$

We adopt the following convention while discussing glued products in the
rest of this article. In the above scenario, we say that $P\times_k Q$
is obtained by gluing $P$ and $Q$ along the $k$ coordinates $I_P$ of $P$ with the
$k$ coordinates $I_Q$ of $Q$. If, for example, these coordinates are named $z$ in
$P$ and $w$ in $Q$, then we also say that $P$ and $Q$ have been glued along the
$z$ and $w$ coordinates and we refer to the coordinates $z$ and
$w$ as the \emph{glued coordinates}. In the special case that we glue
along the last $k$ coordinates, the definition of the glued product simplifies
to
$$\displaystyle P\times_k 
Q=\conv\left(\left\{(x,y,z)\in\mathbb{R}^{d_1+d_2+k}\mid (x,z)\in \vertexset(P),
(y,z)\in \vertexset(Q)\right\}\right).$$

This notion was studied by Margot \cite{Margot_thesis} who provided a sufficient condition
for being able to write the glued product in a specific (and efficient) way from
the descriptions of $P$ and $Q$. We will use this particular way in Lemma
\ref{lem:glued_product}. The existing work \cite{CP12,Margot_thesis}, however, is
more focused on characterizing exactly when this particular method works. We do
not need the result in its full generality and therefore we only state 
a very specific version of it that is relevant for our purposes;
for the sake of completeness, we also provide a proof of it.

\begin{lemma}[Gluing lemma~\cite{Margot_thesis}] \label{lem:glued_product}
	Let $P$ and $Q$ be $0/1$-polytopes and let the $k$ (glued) coordinates 
	in $P$ be labeled $z_1,\ldots,z_k$, and the $k$
	(glued) coordinates in $Q$ be labeled $w_1,\ldots,w_k$. Suppose that
	$\mathbf{1}^\intercal z \leqslant 1$ is valid for $P$ and $\mathbf{1}^\intercal
	w \leqslant 1$ is valid for $Q$. Then $\xc(P\times_k Q)\leqslant
	\xc(P)+\xc(Q)$.
\end{lemma}

\begin{proof}
  Let $(x',z',y',w')$ be a point from $P\times Q \cap \{(x,z,y,w)|z=w\}$. Observe
that the point $(x',z')$ is a convex combination of points $(x',0),
(x',e_1),\ldots, (x',e_k)$ from $P$ with coefficients $(1-\sum_{i=1}^k z'_i),
z'_1,z'_2,\ldots,z'_k$ where $e_i$ is the $i$-th unit vector.  Similarly, the
point $(y',w')$ is a convex combination of points $(y',0), (y',e_1),\ldots,
(y',e_k)$ from $Q$ with coefficients $(1-\sum_{i=1}^k w'_i),
w'_1,w'_2,\ldots,w'_k$. Notice that $(x',0,y')$ as well as $(x',e_j,y')$ for every $j\in \{1, \dots, k\}$, is a
point from the glued product. To see this, let $(x',e_j)=\sum \lambda_v (v,e_j)$ be a convex combination of vertices $(v,e_j)$ of $P$
and $(e_j,y')=\sum\gamma_w (e_j,w)$ be a convex combination of vertices $(e_j,w)$ of $Q$. Clearly, each $(v,e_j,w)$ is a vertex of 
the glued product and $(x',e_j,y')=\sum \lambda_v\gamma_w (v,e_j,w)$ is a convex combination of those vertices. Similarly, $(x',0,y')$ is in the glued product as well.
Now, as $w_i=z_i$ for every $i\in \{0, \dots, k\}$, we conclude
that $(x',w',z')\in P\times_k Q$. Conversely, if $(x',w',z')\in P\times_k Q$ then $(x',z',y',z')\in P\times Q \cap \{(x,z,y,w)|z=w\}.$
Thus, by Proposition \ref{prop:xc_slice} the
extension complexity of $P\times_k Q$ is at most that of $P\times Q$ which is
at most $\xc(P)+\xc(Q)$ by Proposition~\ref{prop:xc_product}.
\end{proof}

The results of the following subsections are needed in Section~\ref{sec:extensions}; they are
not needed for results in Sections~\ref{sec:xc_mso} and~\ref{sec:ef_construction}.


\subsection{Decomposability of Polyhedra}\label{sec:idp_glued}

Now we will define \textit{decomposable} polyhedra and show that decomposability 
is preserved by taking glued product.
Decomposability is also
known as \textit{integer decomposition property} or being
\textit{integrally closed} in the literature (cf. Schrijver~\cite{Schrijver86}). The best known example
are polyhedra given by totally unimodular
matrices~\cite{BT:78}.

  A polyhedron $P \subseteq \RR^n$ is \textit{decomposable} if for
  every $r \in \NN$ and every
  $x \in rP \cap \ZZ^n$, there exist $x^1, \ldots,x^r \in P \cap \ZZ^n$
  with $x = x^1 + \cdots +x^r$, 
  where $rP = \{r y \mid y \in P\}$.
  A \textit{decomposition oracle} for a decomposable $P$ is one that,
  queried on $r \in \NN$ 			
  and on $x \in rP \cap \ZZ^n$, returns $x^1, \ldots,x^r \in
  P \cap \ZZ^n$ with $x = x^1 + \cdots +x^r$.
  If a decomposition oracle for $P$ is realizable by an algorithm
  running in time polynomial in the length of the unary encoding of
  $r$ and $x$,
  we say that $P$ is \textit{constructively decomposable}.

\begin{lemma}[Decomposability and glued product]\label{lem:idp_glued}
  Let $P \subseteq \RR^{d_1 + k}$ and $Q \subseteq \RR^{d_2 + k}$ be $0/1$-polytopes and let the $k$ glued coordinates 
	in $P$ be labeled $z_1,\ldots,z_k$, and the $k$
	glued coordinates in $Q$ be labeled $w_1,\ldots,w_k$. Suppose that
	$\mathbf{1}^\intercal z \leqslant 1$ is valid for $P$ and $\mathbf{1}^\intercal
	w \leqslant 1$ is valid for $Q$.
        Then if $P$ and $Q$ are constructively decomposable,
        so is $P \times_k Q$.
\end{lemma}

\begin{proof}
For the sake of simplicity, we assume without loss of generality 
that glueing is done along the last $k$ coordinates. Then
$P \times_k Q = \conv \{(x,y,z) \in \RR^{d_1 + d_2 + k}\mid
(x,z) \in \vertexset(P), (y,w) \in \vertexset(Q), z=w \}$. Let $R= P \times_k Q$.
To prove that $R$ is constructively decomposable, 
it suffice to find, for every integer $r \in \NN$
and every integer vector $(x,y,z) \in rR$, 
$r$ integer vectors $(x^i, y^i, z^i) \in R$ such that $(x,y,z) =
\sum_{i=1}^r (x^i, y^i, z^i)$. Using the assumption that $P$ and $Q$
are constructively decomposable, we find in polynomial time $r$
integer vectors $(x^i,z^i) \in P$ such that $(x,z) = \sum_{i=1}^r
(x^i, z^i)$ and $r$ integer vectors $(y^j, \bar{z}^j) \in Q$ such that
$(y,z) = \sum_{j=1}^r (y^j, \bar{z}^j)$.

Observe that $z = \sum_{i=1}^r z^i = \sum_{j=1}^r \bar{z}^j$. Moreover, because $z^i$
and $\bar{z}^j$ satisfy $\mathbf{1}^\intercal z^i \leqslant 1$ and
$\mathbf{1}^\intercal \bar{z}^j \leqslant 1$ for all $i$ and $j$, respectively, each vector $z^i$
and each vector $\bar{z}^j$ contains at most one $1$. 
Clearly, the
number of vectors $z^i$ with $z^i_l = 1$ is equal to the number of
vectors $\bar{z}^j$ with $\bar{z}^j_l = 1$, namely $z_l$.

Thus, it is possible to greedily pair the vectors $(x^i, z^i)$ and
$(y^j, \bar{z}^j)$ one to one in such a way that $z^i = \bar{z}^j$
for all the paired vectors. By merging
each such pair of vectors, we obtain $r$ new integer vectors $(x^l, y^l, z^l) \in
R$, for $1 \leq l \leq r$, that satisfy $(x,y,z) = \sum_{l=1}^r (x^l,
y^l, z^l)$, concluding the proof.
\end{proof}

The following lemma will be useful:

\begin{lemma}\label{lem:idp_projection}
  Let $Q \subseteq \RR^{n'}$ be a polyhedron which is constructively
  decomposable and let $\pi: \RR^{n'} \rightarrow \RR^n$ be a linear projection with integer
  coefficients. Then the polytope $R=\conv(\{(y,\pi(y)) \mid y
  \in Q\})$ is constructively decomposable.
\end{lemma}

\begin{proof}
  Consider an integer $r$ and an integer vector $(y,x) \in rR$. Since $Q$ is
  constructively decomposable, we can find in polynomial time $r$ vectors
  $y^i \in Q\cap \ZZ^{n'}$, for $1 \leq i \leq r$, such that $y = \sum_{i=1}^r  y^i$. For
  every $i$, let $x^i = \pi(y^i)$; note that every $x^i$ is integral. 
  Since $x = \pi(y) = \sum_i \pi(y^i) =
  \sum_i x^i$, we conclude that $(y,x)=\sum_{i=1}^r(y^i, x^i)$,
  proving that $R$ is constructively decomposable.
  \end{proof}

Obviously, not all integer polyhedra are decomposable: consider the
three-dimensional parity polytope $P = \conv(\{(0,0,0), (1,1,0), (1,0,1), (0,1,1)\})$
and the point $(1,1,1) \in 2P$ -- there is no way to express it as 
a sum of integral points in $P$. However, the following lemma shows that
every integer polyhedron has an extension that is decomposable.

\begin{lemma}\label{lem:decomposable_extension}
Every $0/1$ integer polytope $P$ has an extension $R$ that is constructively decomposable.
Moreover, a description of such an extension can be computed in time
$\Oh(d+n)$ where $d$ is the dimension of $P$ and $n$ is the number of
vertices of $P$ if the vertices of $P$ are given. 
\end{lemma}

\begin{proof}
Let $\vertexset(P) = \{v_1,\ldots,v_n\}$ denote all the vertices of $P$ and  
let $Q=\{\lambda \mid \sum_{i=1}^n \lambda_i =1, \lambda_i\geq 0 
\mbox{ for } i\in\{1,\ldots,n\} \}$ be the $n$-dimensional simplex.
Then, for the linear projection $\pi(\lambda)=\sum_{i=1}^{n}\lambda_i v_i$, 
the polytope $R=\{(\pi(\lambda),\lambda)  \mid  \lambda \in Q \}$ is an
extended formulation of $P$
(note that the same extended formulation of $P$ is used 
also in the proof of Proposition~\ref{prop:xc_vertices}).

Consider an arbitrary integer $r$ and an integral point
$x = (x_1, \dots, x_n) \in rQ$. As $x = \sum_{j=1}^n x_j
e_j$ where $e_j$ is the $j$-th unit vector,
and as $\sum_{j=1}^n x_j=r$, we see that $x$ 
can be written as a sum of at most $r$ integral points from $Q$,
and such a decomposition can be found in time polynomial in $n$ and $d$.
Thus, $Q$ is constructively decomposable. 
Then, applying the previous lemma to the simplex $Q$ and 
the linear projection $\pi$, we see that the polytope $R$
is constructively decomposable.
\end{proof}
Given a polytope $P$, it is an interesting problem to determine the
minimum size of an extension of $P$ that is
decomposable. This is an analogue of extension complexity:
the \textit{decomposable extension complexity} of a polytope $P$,
denoted $\xcdec(P)$, is the minimum size of an extension of $P$
that is decomposable. A polytope $Q$ which is an extension
of $P$ and is decomposable is called a \textit{decomposable extension
  of $P$}.

Obviously, $\xc(P) \leq \xcdec(P)$. Using
Lemma~\ref{lem:idp_projection} and Proposition~\ref{prop:xc_vertices}
we see that if a polytope $P$ has $n$ vertices, then $\xcdec(P) \leq n$. It is an
interesting problem to determine for which polytopes $\xc(P) =
\xcdec(P)$, or, on the other hand, when $\xc(P) < \xcdec(P)$ and by
how much they can differ.

\subsection{Treewidth of Gaifman Graphs of Extended Formulations}\label{sec:tw_glued}

  Given a relational structure $\mathcal{A} = (A, \emptyset, \mathcal{S})$ where $\mathcal{S}
  \subseteq 2^A$, its \textit{Gaifman graph} is the graph
  $G(\mathcal{A}) = (A, E)$ where $E = \{\{u,v\}\mid\exists S \in
  \mathcal{S}: u,v \in S\}$.
The \textit{Gaifman graph} $G(A)$ associated with a matrix $A \in \RR^{m\times
  n}$ is the Gaifman graph of the structure
$(\{1,\dots,n\},\emptyset, \{\suppo(a_i)\mid1 \leq i \leq m\})$ where $a_i$ is the $i$-th
row of $A$ and $\suppo(x)$ is
the \textit{support} of a vector $x$, that is, the set of indices $i$
such that $x_i \neq 0$. In other words, the graph $G(A)$ has a vertex
for each column of $A$ and two vertices are connected by an edge if the
supports of the corresponding columns have non-empty intersection.

The \textit{treewidth of a matrix} $A \in \RR^{n \times m}$, denoted
$tw(A)$, is the treewidth of its Gaifman graph. 
The \textit{treewidth of a system of inequalities} $Ax\leq b$
is defined as $tw(A)$.
Since each graph $G$ has a trivial tree decomposition of width $|V(G)|$ which puts all vertices into the bag of a single node, clearly $tw(A) \leq n$.


Note that in the following lemma, the meaning of the variables $x$ and $y$ is different than before: $(x,z)$ are the variables not of $P$ but some extended formulation of $P$, and analogously for $(y,w)$ and $Q$.
\begin{lemma}[Treewidth and glued product]\label{lem:tw_glued}
Let $P$ and $Q$ be $0/1$-polytopes and let the $k$ glued coordinates in $P$ be
labeled $z_1,\ldots,z_k$ and the $k$ glued coordinates in $Q$ be labeled
$w_1,\ldots,w_k$. 
Suppose that $\mathbf{1}^\intercal z \leqslant 1$ is valid for $P$ and
$\mathbf{1}^\intercal w \leqslant 1$ is valid for $Q$.
Let $Ax + Cz \geq a$ be inequalities describing an extended formulation of $P$ and $Dw
+ Ey \geq b$ be inequalities describing an extended formulation of $Q$.
Let $F = \bigl(\begin{smallmatrix} A&C&0
	\\ 0&D&E \end{smallmatrix} \bigr)$ and $c = \bigl(\begin{smallmatrix} a
	\\ c \end{smallmatrix} \bigr)$.
Then the polytope described by $F (x,y,z) \geq c$ is an extended formulation of $P \times_k Q$ and $tw(F) \leq \max\{tw(A~C), tw(D~E), k\}$.

Moreover, if $(T_P, B_P)$ is a tree decomposition of $G(A~C)$ of treewidth $tw(A~C)$
with a node $d$ with ``columns of $C$'' $\subseteq B_P(d)$ and $(T_Q, B_Q)$ is a tree
decomposition of $G(D~E)$ of treewidth $tw(D~E)$ containing a node $d'$ with
``columns of $D$'' $\supseteq B_Q(d')$, then there is a tree decomposition $(T_R, B_R)$ of $G(F)$ of treewidth
$\max\{tw(A~C), tw(D~E), k\}$, where $T_R$ is obtained from $T_P$ and $T_Q$ by
identifying the nodes $d$ and $d'$ and $B_R = B_P \cup (B_Q \setminus \{B_Q(d')\})$.
\end{lemma}

\begin{proof}
We start by observing that the assumptions and the gluing lemma imply that the inequalities
  \begin{alignat*}{2}
    Ax + &Cz &&\geq a \\
    &Dz + Ey &&\geq b
  \end{alignat*}
describe an extended formulation of $P \times_k Q$. Consider the treewidth
  of the matrix $F = \bigl(\begin{smallmatrix} A&C&0
    \\ 0&D&E \end{smallmatrix} \bigr)$. The Gaifman graph $G(F)$ of
  $F$ can be obtained by taking $G(A~C)$ and $G(D~E)$ and identifying the
  vertices corresponding to the variables $z$ and $w$ in the above formulation. 
  It is easy to observe
  that the treewidth of $G(F)$ is $\max(tw(P), tw(Q), k)$,
  as desired, and that if $(T_P, B_P)$ and $(T_Q, B_Q)$ are as assumed,
  the tuple $(T_R, B_R)$ obtained in the aforementioned way
  is indeed a tree decomposition of $G(F)$.
\end{proof}

\begin{lemma}\label{lem:tw_vertices}
Let $P\subseteq \RR^m$ be a polytope with $n$ vertices $v_1,\ldots,v_n$. Then there exist
an extension of $P$ that can be described by inequalities of treewidth at most $n + m$.
\end{lemma}

\begin{proof} 
Consider again the description of the extension of $P$ used in the proof of Proposition~\ref{prop:xc_vertices}:
  \begin{alignat}{4}\label{eqn:tw_vertices_start}
    \mathbf{1} & \lambda~ ~ &&  &= &\ 1 \\
    V & \lambda~ -~ && Ix\  &= &\ \mathbf{0} \\
    & \lambda & & & \geq &\ \mathbf{0}\label{eqn:tw_vertices_end}
  \end{alignat}
 where $\mathbf{0}$ and $\mathbf{1}$ are the all-0 and all-1 vectors
 of appropriate dimensions, respectively, $I$ is the identity
 matrix and $V$ is a matrix whose $i$-th column is $v_i$
 (each equality is just an abbreviation of two opposing inequalities).
 Since the number of columns in the system is $n+m$, its treewidth is by definition
 also at most $n+m$.
\end{proof}

\medskip

Putting Lemmas~\ref{lem:decomposable_extension} and~\ref{lem:tw_vertices} together gives the following corollary:

\begin{corollary}\label{cor:simplex_idp_tw}
Let $P\subseteq \RR^m$ be an integral polytope with $n$ vertices $v_1,\ldots,v_n$.
Then there exist
a constructively decomposable extension of $P$ that can be described by inequalities of treewidth at most $n + m$.
\end{corollary}
\begin{proof}
It suffices to notice that the extended formulations of Lemmas~\ref{lem:decomposable_extension} and~\ref{lem:tw_vertices} are identical and thus simultaneously have small treewidth and are decomposable. 
\end{proof}

\section{Extension Complexity of $P_\varphi(G)$} \label{sec:xc_mso}

For a given MSO$_2$ formula $\varphi(\vec{X})$ 
with $m$ free set variables $X_1,\ldots,X_m$, we define a
polytope of satisfying assignments of the formula $\varphi$ on a given graph $G$, represented
as a $\sigma_2$-structure $I(G)=(V_I,\emptyset,\{E_I,L_V,L_E\})$ with domain of size $n = |V_I| = |V(G)| + |E(G)|$, in a natural way.
We encode any assignment of elements of $I(G)$ to the sets
$X_1,\ldots,X_m$
as follows.
For each $X_i$ in $\varphi$ and each $v$ in $V_I$, we introduce a binary variable $y_v^i$.
We set $y_v^i$ to be one if $v\in X_i$ and zero otherwise.
For a given $0/1$ vector $y$, we say that $y$ {\em satisfies} $\varphi$ if interpreting the coordinates of $y$ as described above yields a satisfying assignment for $\varphi$. The polytope of satisfying assignments of the formula $\phi$ on the graph $G$ is defined as

$$\displaystyle P_\varphi(G) = \conv\left(\left\{y \in \{0,1\}^{nm} \ |\ y
\text{
satisfies } \varphi \right\}\right). \ $$

As an example, consider again the formula $\varphi_{\textrm{3-col}}$ whose satisfying assignments are valid $3$-colorings.
Take $G$ to be the path on $3$ vertices.
Its representation as $I(G)$ has a universe $V_I = \{v_1, v_2, v_3, e_1, e_2\}$, the labels are $L_V = \{v_1, v_2, v_3\}$ and $L_E = \{e_1, e_2\}$, and the binary relation $E_I$ is $\{v_1 e_1, e_1 v_2, v_2 e_2, e_2 v_3\}$.
The two interesting colorings of $G$ by red, green, and blue are $RGB$ and $RGR$, and permuting the colors gives $10$ additional isomorphic colorings ($RBG, GRB, GBR, BRG, BGR$ using three colors, and $RBR, GRG, GBG, BRB, BGB$ using two colors).
The polytope $P_\varphi(G)$ has dimension $3\cdot |V_I| = 15$.
For example, the coloring $RGB$ is encoded by a vertex $y$ of $P_{\varphi_{\textrm{3-col}}}(G)$ which has $y_{v_1}^R = 1, y_{v_2}^G = 1, y_{v_3}^B = 1$ and all other coordinates are zero.
Thus, $P_{\varphi_{\textrm{3-col}}}(G)$ is the convex hull of the $12$ aforementioned $3$-colorings of $G$.

For the sake of simplicity, we state the following theorem and carry out the exposition for {\em graphs}; however, identical arguments can be carried out analogously for a $\sigma$-structure whose Gaifman graph has treewidth bounded by $\tau$ for any fixed vocabulary $\sigma$.
\begin{theorem}[Extension Complexity of $P_\varphi(G)$]
\label{thm:polytope_courcelle}
For every graph $G$ represented as a $\sigma_2$-structure $I(G)$
and for every MSO$_2$ formula $\varphi$,
$$\xc(P_\varphi(G)) \leq f(|\varphi|, \tau)\cdot n,$$ where $f$ is some computable
function, $\tau=tw(G)$ and $n=|V_I| = |V(G)| + |E(G)|$.
\end{theorem}
\begin{proof}
Let $(T, B)$ be a fixed nice tree decomposition of treewidth $\tau$ of
$I(G)$ and let $k$ denote the quantifier rank of $\varphi$
and $m$ the number of free variables of $\varphi$. 
Recall that $\CCtw$ is the set of equivalence classes of the relation $\equivMSOkt$ of treewidth bounded by $\tau$.
For each node $b$ of $T$ we introduce $|\CCtw|$ binary variables that will 
represent a feasible type of the node $b$; we denote the vector
of them by $t_b$ (i.e., $t_b\in\{0,1\}^{|\CCtw|}$).
For each introduce and each forget node $b$ of $T$,
we introduce additional $|\CCtw|$ binary variables that will represent
a feasible type of the child (descendant) of $b$; we denote the vector
of them by $d_b$ (i.e., $d_b\in\{0,1\}^{|\CCtw|}$).
Similarly, for each join node $b$ we introduce additional $|\CCtw|$ 
binary variables, denoted by $l_b$, that will represent a feasible type of the
left child 
of $b$, 
and other $|\CCtw|$ binary variables, denoted by $r_b$, that will represent a
feasible type 
of the right child of $b$
(i.e., $l_b, r_b\in\{0,1\}^{|\CCtw|})$.

We are going to describe inductively a polytope in the dimension given (roughly)
by all the binary variables of all nodes of the given nice tree decomposition.
Then we show that its extension complexity is small and that a properly chosen
face of it is an extension of $P_\varphi(G)$.

First, for each node $b$ of $T$, depending on its type, we define a polytope
$P_b$ as follows:
\begin{itemize}
  \item $b$ is a \textit{leaf}. $P_b$ consists of a single point
$P_b=\{\overbrace{1 0 0 \ldots 0}^{|\CCtw|}\}$.
  \item $b$ is an \textit{introduce} or \textit{forget} node. For each feasible
pair
of types $(\alpha_i,\alpha_j)\in \FF_p(b)$ of the node $b$, 
we create a vector $(d_b,t_b)\in
\{0,1\}^{2|\CCtw|}$ with $d_b[i]=t_b[j]=1$ 
and all other coordinates zero. $P_b$ is defined as the
convex hull of all such vectors.
   \item $b$ is a \textit{join} node. For each feasible triple of types
$(\alpha_h,\alpha_i,\alpha_j)\in \FF_t(b)$ of the node $b$, we create a vector
$(l_b,r_b,t_b)\in \{0,1\}^{3|\CCtw|}$
with $l_b[h]=r_b[i]=t_b[j]=1$ and all other coordinates zero. 
$P_b$ is defined as the convex hull of all such vectors.
\end{itemize} 

It is clear that for every node $b$ in $T$, the polytope $P_b$
contains at most $|\CCtw|^3$ vertices, and, thus, by
Proposition~\ref{prop:xc_vertices} it has extension complexity at most
$\xc(P_b)\leqslant |\CCtw|^3$. Recalling our discussion in
Section~\ref{sec:prelim} about the size of $\CCtw$, we conclude that there exists
a function $f$ such that for every $b\in V(T)$, it holds that
$\xc(P_b)\leqslant f(|\varphi|,\tau)$.

We create an extended formulation for
$P_\varphi(G)$ by gluing these polytopes together, starting in the leaves of $T$
and processing $T$ in a bottom up fashion. 
We create polytopes $Q_b$ for each node $b$ in $T$
recursively as follows:
\begin{itemize}
\item If $b$ is a leaf then $Q_b=P_b$. 
\item If $b$ is an introduce or forget node, then 
$Q_b=Q_a\times_{|\CCtw|}P_b$ where $a$ is the child of $b$ and the gluing is done
along the coordinates
$t_a$ in $Q_a$ and $d_b$ in $P_b$. 
\item If $b$ is a join node, then we first define $R_b = Q_a\times_{|\CCtw|}P_b$
where
$a$ is the left child of $b$ and the gluing is done along the coordinates $t_a$
in $Q_a$ and $l_b$ in $P_b$. Then $Q_b$ is obtained by gluing $R_b$
with $Q_c$ along the coordinates $t_c$ in $Q_c$ and $r_b$ in $R_b$ where $c$ is
the right child of $b$.
\end{itemize}

The following lemma states the key property of the polytopes $Q_b$.
\begin{lemma}\label{lem:vertices_labelings}
For every vertex $y$ of the polytope $Q_b$ there exist 
$X_1,\ldots,X_m\subseteq V(G_b)$ such
that $(G_b, (X_1,\ldots,X_m), \eta(B(b)))$ is of type $\alpha$ where
$\alpha$ is the unique type such that
the coordinate of $y$ corresponding to the binary variable $t_b[\alpha]$ 
is equal to one. 
\end{lemma}


\begin{proof}
The proof is by induction, starting in the leaves of $T$ and going up
towards the root. For leaves, the lemma easily follows from the definition
of the polytopes $P_b$.

For the inductive step, we consider an inner node $b$ of $T$ and we distinguish 
two cases:
\begin{itemize}
\item If $b$ is a join node, then the claim for $b$ follows from the
inductive assumptions for the children of $b$, definition of a feasible
triple, definition of the polytope $P_b$, Lemma~\ref{lem:only_types_matter}
and the construction of the polytope $Q_b$.
\item If $b$ is an introduce node or a forget node, respectively, then, 
analogously, the claim for $b$ follows from the
inductive assumption for the child of $b$, definition of a feasible
pair, definition of the polytope $P_b$, Lemma~\ref{lem:introduce_matter}
or Lemma~\ref{lem:forget_matter}, respectively,
and the construction of the polytope $Q_b$.
\end{itemize}
\end{proof}

Let $c$ be the root node of the tree decomposition $T$. Consider the polytope
$Q_{c}$. From the construction of $Q_{c}$, our previous discussion and the 
Gluing lemma, it follows 
that $\xc(Q_{c})\leqslant \sum_{b\in V(T)} \xc(P_b)\leqslant f(|\varphi|,
\tau)\cdot n$. It remains to show that a properly chosen
face of $Q_c$ is an extension of $P_{\varphi}(G)$.
We start by observing that $\sum_{i=1}^{|\CCtw|}
t_{c}[i] \leq 1$ and $\sum_{i=1}^{|\CCtw|}
\rho_\varphi(i){\cdot}t_{c}[i] \leq 1$, where $\rho_\varphi$ is 
the indicator function defined in Subsection~\ref{subsec:MSO_types}, 
are valid inequalities for $Q_c$.

Let $Q_\varphi$ be the face of $Q_c$
corresponding to the valid inequality $\sum_{i=1}^{|\CCtw|}
\rho_\varphi(i){\cdot}t_{c}[i] \leq 1$. 
Then, by Lemma~\ref{lem:vertices_labelings}, the polytope 
$Q_\varphi$ represents those
$[m]$-labellings of $G$ for which $\varphi$ holds.
The corresponding feasible assignments of $\varphi$ on $G$ are
obtained as follows: 
for every vertex $v\in V(G)$ and every $i\in \{1,\ldots,m\}$ we set
$y_v^i=\sum_{j=1}^{|\CCtw|}
\mu(\alpha_j,v,i){\cdot}t_{top(v)}[j]$.
The sum is $1$ if and only if there exists a type $j$ such that
$t_{top(v)}[j]=1$ and at the same time $\mu(\alpha_j,v,i)=1$; by the definition
of the indicator function $\mu$ in Subsection~\ref{subsec:MSO_treewidth}, this
implies that $v\in
X_i$. 
Thus, by applying the above projection to $Q_\varphi$ we obtain
$P_\varphi(G)$, as desired.  

It is worth mentioning at this point that the polytope $Q_c$ depends
only on the treewidth $\tau$, the quantifier rank $k$ of $\varphi$
and the number of free variables of $\varphi$.
The dependence on the formula $\varphi$
itself only manifests in the choice of the face $Q_\varphi$ of $Q_c$ and
its projection to $P_\varphi(G)$.
\end{proof}

\section{Efficient Construction of the $P_\varphi(G)$}\label{sec:ef_construction}

In the previous section we have proven that $P_\varphi(G)$
\textit{has} a compact extended formulation but our definition of
feasible tuples and the indicator functions $\mu$ and $\rho_\varphi$
did not explicitly provide a way how to actually \textit{obtain} it
efficiently.
That is what we do in this section.


As in the previous section we assume that we are given a graph $G$ of treewidth
$\tau$ and an MSO$_2$ formula $\varphi$ with $m$ free variables and quantifier rank
$k$. We start by constructing a nice tree decomposition $(T, B)$ of $G$ of treewidth
$\tau$ in time $f(\tau) \cdot n$~\cite{Bodlaender:93,Kloks:94}.

Recall that $\CCtw$ denotes the set of equivalence classes of $\equivMSOkt$ which have treewidth at most $\tau$.
Observe that we can restrict our attention from $\CC$ to $\CCtw$ because any subgraph of $G$ has treewidth at most $\tau$ and thus the types in $\CC \setminus \CCtw$ are not feasible for any node of $T$.
 Because $\CC$
is finite and its size is independent of the size of $G$ (Corollary
\ref{cor:finite_index}), so is $\CCtw$, and for each class $\alpha\in \CCtw$ there exists 
an $[m]$-labelled $\tau$-boundaried graph $(G^\alpha,
\vec{X}^\alpha, \vec{p^\alpha})$ of type $\alpha$ whose size is upper-bounded
by 
a function of $k,m$ and $\tau$.
For each $\alpha\in \CCtw$, we fix one
such graph, denote it by $W(\alpha)$ and call it the \textit{witness of
$\alpha$}. Let $\WW = \{W(\alpha) \mid \alpha \in \CCtw\}$. The witnesses
make it possible to easily compute the indicator function $\rho_\varphi$:
for every $\alpha \in \CCtw$, we set $\rho_\varphi(\alpha) = 1$ if and only if 
$W(\alpha) \models \varphi$, and we set $\rho_\varphi(\alpha) = 0$ otherwise.

The following Lemma is implicit in~\cite{GPW:07} in the proof of Theorem 4.6 and
Corollary 4.7.

\begin{lemma}[\cite{GPW:07}] \label{lemma:compute_type_witnesses} The set $\WW$ and
  the indicator function $\rho_\varphi$ can be computed in time $f(k, m,
  \tau)$, for some computable function~$f$.
\end{lemma}

It will be important to have an efficient algorithmic test for
$\MSO[k,\tau]$-elementary equivalence. This can be done using the \EF
games:

\begin{lemma}[{\cite[Theorem~7.7]{Libkin:04}}] \label{lem:ef_games}
  Given two $[m]$-labelled $\tau$-boundaried graphs $G_1^{[m],\tau}$
  and $G_2^{[m],\tau}$, it can be decided in time $f(m, k, \tau, |G_1|, |G_2|)$ 
  whether $G_1^{[m],\tau} \equivMSOkt G_2^{[m],\tau}$, for some computable
function~$f$.
  \end{lemma}
\begin{corollary} \label{cor:recog} 
Recognizing the type of an $[m]$-labelled $\tau$-boundaried graph
$G^{[m],\tau}$ can be done in time $f(m, k, \tau, |G|)$, for some computable
function~$f$.  
\end{corollary}


Now we describe a linear time construction of the sets of feasible types, pairs
and triples of 
types $\FF(b)$, $\FF_p(b)$ and $\FF_t(b)$ for all
relevant nodes $b$ in $T$.
In the initialization phase we construct the set $\WW$, using the
algorithm from Lemma~\ref{lemma:compute_type_witnesses}. 
The rest of the construction is inductive, starting in the leaves
of $T$ and advancing in a bottom up fashion towards the root of $T$.
The idea is to always replace a possibly {\em large} graph
$G_a^{[m], \tau}$ of type $\alpha$ by the {\em small} witness $W(\alpha)$
when computing the set of feasible types for the father of a
node $a$.

\textit{Leaf node.} For every leaf node $a \in V(T)$ we set $\FF(a) =
\{\alpha_1\}$. Obviously, this corresponds to the definition in 
Section~\ref{sec:prelim}.

\textit{Introduce node.} Assume that $b \in V(T)$ is an introduce node with a
child $a \in V(T)$ for which $\FF(a)$ has already been computed,
and $v\in V(G)$ is the new vertex.
For every $\alpha \in \FF(a)$, we first produce a $\tau'$-boundaried
graph $H^{\tau'}=(H^\alpha, \vec{q})$ from 
$W(\alpha)=(G^\alpha, \vec{X^\alpha}, \vec{p^\alpha}) $ as follows: let
$\tau' = |\vec{p^\alpha}|+1$ and $H^\alpha$ be obtained from $G^\alpha$ by
attaching 
to it a new vertex in the same way as $v$ is attached to $G_a$. The boundary
$\vec{q}$ is obtained from the boundary  $\vec{p^\alpha}$ by inserting in it
the 
new vertex at the same position that $v$ has in the boundary of
$(G_a,\eta(B(a)))$. 
For every subset $I\subseteq \{1,\ldots,m\}$ we construct
an $[m]$-labelling $\vec{Y}^{\alpha,I}$ from $\vec{X^\alpha}$ by setting 
${Y}_i^{\alpha,I} = X_i^\alpha\cup \{v\}$, for every $i\in I$,  and 
${Y}_i^{\alpha,I} = X_i^\alpha$, for every $i\not \in I$. Each of
these $[m]$-labellings $\vec{Y}^{\alpha,I}$ is used to produce an $[m]$-labelled
$\tau'$-boundaried graph $(H^\alpha, \vec{Y}^{\alpha,I}, \vec{q})$ and the types
of
all these $[m]$-labelled $\tau'$-boundaried graphs are added to the set $\FF(b)$
of feasible types of $b$, and, similarly, the pairs $(\alpha, \beta)$ where
$\beta$ is a feasible type of some of the $[m]$-labelled $\tau'$-boundaried
graph $(H^\alpha, \vec{Y}^{\alpha,I}, \vec{q})$, are added to the set $\FF_p(b)$
of
all feasible pairs of types of $b$.  
The correctness of the construction of the sets $\FF(b)$ and $\FF_p(b)$
for the node $b$ of $T$ follows from Lemma~\ref{lem:introduce_matter}.

\textit{Forget node.} Assume that $b \in V(T)$ is a forget node with a
child $a \in V(T)$ for which $\FF(a)$ has already been computed
and that the $d$-th vertex of the boundary $\eta(B(a))$
is the vertex being forgotten.
We proceed in a similar way as in the case of the introduce node. For
every $\alpha \in \FF(a)$ we produce an $[m]$-labelled $\tau'$-boundaried 
graph $(H^\alpha, \vec{Y}^\alpha, \vec{q})$
from $W(\alpha) = (G^\alpha, \vec{X}^\alpha, \vec{p}^\alpha)$ as follows: 
let $\tau' = |\vec{p^\alpha}|-1$, $H^{\alpha} = G^{\alpha}$,
$\vec{Y}^\alpha=\vec{X}^\alpha$
and $\vec{q} = (p_1, \dots, p_ {d-1}, p_{d+1}, \dots,
p_{\tau'+1})$. For every $\alpha\in \FF(a)$, the type $\beta$ of the constructed
graph is added to $\FF(b)$, and, similarly, the pairs $(\alpha,\beta)$
are added to $\FF_p(b)$.
The correctness of the construction of the sets $\FF(b)$ and $\FF_p(b)$
for the node $b$ of $T$ follows from Lemma~\ref{lem:forget_matter}.

\textit{Join node.} 
Assume that $c \in V(T)$ is a join node with children $a,b \in V(T)$ for which
$\FF(a)$ and $\FF(b)$ have already been computed.
For every pair of compatible types $\alpha \in \FF(a)$ and
$\beta \in \FF(b)$, we add the type $\gamma$ of $W(\alpha) \oplus W(\beta)$ 
to $\FF(c)$, and the triple $(\alpha, \beta, \gamma)$ to $\FF_t(c)$.
The correctness of the construction of the sets $\FF(c)$ and $\FF_t(c)$
for the node $b$ of $T$ follows from Lemma~\ref{lem:only_types_matter}.

It remains to construct the indicator functions $\nu$ and $\mu$.
We do it during the construction of the sets of feasible types as follows. 
We initialize $\nu$ to zero. Then, every
time we process a node $b$ in $T$ and we find a new feasible type $\beta$ of
$b$, for every $v\in B(b)$ and for every $i$ for which the $d$-th vertex in the
boundary of $W(\beta)=(G^\beta,\vec X,\vec p)$ belongs to $X_i$, we set
$\mu(\beta,b,v,i)=1$ where $d$ is the order of $v$ in the boundary of
$(G_b,\eta(B(b))$. The correctness follows from the definition of $\nu$
and the definition of feasible types.
The function $\mu$ is then straightforwardly defined using $\nu$.

Concerning the time complexity of the inductive construction, we observe,
exploiting Corollary~\ref{cor:recog},
that for every node $b$ in $T$, the number of steps, the sizes of 
graphs that we worked with when dealing with the node $b$,
and the time needed for each of the steps, depend on 
$k$, $m$ and $\tau$ only.
We summarize the main result of this section in the following theorem.



\begin{theorem} \label{thm:constructive_polytope_courcelle}
Under the assumptions of Theorem~\ref{thm:polytope_courcelle}, the polytope
$P_\varphi(G)$ can be constructed in time $f'(|\varphi|, \tau)\cdot n$, for some
computable function $f'$.
\end{theorem}

\section{Extensions}\label{sec:extensions}

Using our results about the glued product in Section~\ref{sec:glued_product} we
can extend Theorem~\ref{thm:constructive_polytope_courcelle} to guarantee
a couple of additional non-trivial properties of
a certain extended formulation of $P_\varphi(G)$.
Recall that $\CCtw=\{\alpha_1,\ldots,\alpha_w\}$ is the set of equivalence classes of the relation $\equivMSOkt$ where for each $G^{[m], \tau} \in \alpha \in \CCtw$, we have $tw(G^{[m], \tau}) \leq \tau$.
For a given formula $\varphi$, a graph $G$ and a tree decomposition $(T, B)$ of $G$, for every node $a$ of $T$, we denote the set of 
feasible types of the node $a$ by $\mathcal{F}(a)$, for every introduce and every forget node $a$ of $T$ the set of feasible pairs of the node $a$ by $\FF_p(a)$ and for every join node $a$ of $T$ the set of feasible triples of the node $a$ by $\FF_t(a)$.
Moreover, 
let $V_{IF}$ denote the set of introduce and forget nodes in $T$, $V_J$ the set of join nodes and $V_L$ the set of leaves,
and
let $\mathcal{F}=\bigcup_{b\in V_{IF}} \{\{b\}\times \FF_p(b)\} \cup
\bigcup_{b\in V_{J}} \{\{b\}\times \FF_t(b)\} \cup
\bigcup_{b\in V_{L}} \{\{b\}\times \FF(b)\}$,
that is, $\mathcal{F}$ is a set containing for every node $b\in V(T)$ a pair $\{b,\FF'(b)\}$ 
where $\FF'(b)$ is the set of feasible pairs $\FF_p(b)$ for introduce and forget nodes, 
the set of feasible triples $\FF_t(b)$ for join nodes and
the set of feasible types $\FF(b)$ for leaves. 

As in the case of Theorem~\ref{thm:polytope_courcelle}, for the sake of simplicity we again formulate and prove the main theorem of this section in terms of graphs represented as $\sigma_2$-structures; the extension to arbitrary structures is straightforward.
\begin{theorem} \label{thm:master_courcelle}
  Let $G=(V, \emptyset, \{E, L_V, L_E\})$ be a $\sigma_2$-structure of treewidth $\tau$ representing a graph, let $n=|V|$ and let 
$(T, B)$ be a nice tree
decomposition of $G$ of treewidth $\tau$ and let $\varphi$ be an MSO$_2$ formula with $m$ free variables.


  Then there exist matrices $A, D$, $C$, a vector $e$, 
  a function $\nu: \CCtw \times V(T)\times V \times \{1,\ldots,m\} \rightarrow \{0,1\}$
  and a tree decomposition $(T^*, B^*)$
  of the Gaifman graph $G(A~D~C)$
  such that the following claims hold:
      \begin{enumerate}
  \item \label{thm:master_courcelle:polytope}The polytope $P = \{(y,t,f) \in \RR^{V\times [m]}\times \RR^{\CCtw \times V(T)} \times \RR^{\FF} \mid Ay + Dt + Cf = e, \quad t,f \geq
    \mathbf{0}\}$ is a $0/1$-polytope and $P_\varphi(G) = \{y \mid \exists t,f: (y,t,f) \in P\}$.
  \item \label{thm:decomp} $P$ is constructively decomposable.
  \item \label{thm:variables}
    For any $(y,t,f) \in \vertexset(P)$,
    for any $j \in \CCtw$, $b \in V(T)$, $v \in B(b)$
    and $i \in \{1, \dots, m\}$, equalities $t_b[j] = 1$ and 
    $\nu(j,b,v,i) = 1$ imply that $y_v^i = 1$.
  \item \label{thm:tw_structure} 
    \begin{enumerate}
      \item \label{thm:tw_structure:bounded}The treewidth of $(T^*, B^*)$ is $\Oh(|\CCtw|^3)$,
      \item $T^* = T$,
      \item \label{thm:tw_structure:bags}for every node $b \in V(T^*)$, $\bigcup_{j\in \CCtw}\{t_b[j]\} \subseteq B^*(b)$, and,
      \item $\bigcup_{j\in \CCtw} \{t_b[j]\} \cap B^*(a) = \emptyset$ for every $a
        \not\in \delta_{T^*}(b)$.
     \end{enumerate}
  \item \label{thm:adcdnu_computable} $A, D, C, d, \nu$ can be computed in time $\Oh(|\CCtw|^3\cdot n)$.
  \end{enumerate}
\end{theorem}

Let us first comment on the meaning and usefulness of the various points of the theorem.
Point~\eqref{thm:master_courcelle:polytope} simply states that $P$ is an extended formulation of $P_\varphi(G)$.
However, there are variables $t$ which allow some interpretation of integer points of $P$ as we will discuss further (point~\eqref{thm:variables}), and there are also variables $f$, which are used to ensure constructive decomposability (point~\eqref{thm:decomp}).

There are currently two applications of this theorem, one due to Gajarský et al.~\cite{GajarskyHKO:17} and the other by Knop et al.~\cite{KnopKMT:17}.
What they have in common is viewing the system of linear inequalities $Ay + Dt + Cf = e$ with $t, f \geq 0$, which defines $P$, as an integer linear program, because they are only interested in its integer solutions, and then further viewing it as a constraint satisfaction problem (CSP).
This allows adding nonlinear constraints and optimizing nonconvex objective function.
Then, by point~\eqref{thm:tw_structure:bounded}, this CSP instance also has bounded treewidth, and thus can be efficiently solved by an old algorithm of Freuder~\cite{Freuder:90}.

Point~\eqref{thm:variables} is closely related to Lemma~\ref{lem:vertices_labelings} which we needed for the proof of Theorem~\ref{thm:polytope_courcelle}.
Intuitively, it says that we can view the variables $t$ of integer points of $P$ as an assignment from $V(T)$ to $\CCtw$ (i.e., each node is assigned a type) and that knowing a type of a node $b$ is sufficient for knowing, for each vertex $v \in B(b)$, to which free variables $X_i$ vertex $v$ belongs.
This was used by Knop et al.~\cite{KnopKMT:17} who study various extensions of the MSO logic.
In their work, they extend a CSP instance corresponding to the system defining $P$ with further constraints modeling the various extensions of MSO.
They crucially rely on point~\eqref{thm:tw_structure:bags} which allows them to add new constraints in a way which does not increase the treewidth of the resulting CSP instance by much.

Point~\eqref{thm:decomp} essentially says that integer points of $rP$ correspond to $r$-sets (i.e., sets of size $r$) of vertices of $P$.
This was used by Gajarský et al.~\cite{GajarskyHKO:17} who study the so-called shifted combinatorial optimization problem (SCO), where one wants to optimize a certain non-linear objective over $r$-sets of a given set $S$.
Gajarský et al. connect separable optimization over the $r$-dilate of a decomposable $0/1$ polyhedron $Q$ with SCO.
Thus, when $S = S_\varphi(G)$ is the set of satisfying assignments of a formula $\varphi$ on a graph $G$, one can optimize over integer points of $rP$ in order to optimize over $r$-sets of $S_\varphi(G)$.


\begin{proof}[of Theorem~\ref{thm:master_courcelle}]
%
Let us give an outline of the proof first.
The construction of the polytope $P$, and of the corresponding system of 
linear inequalities describing it, is done in three phases.
In each phase we construct and examine three related objects: a certain
polytope, 
a system of linear inequalities
defining it, and a tree decomposition of the Gaifman graph of the
system of linear inequalities.
We closely follow along the lines of the proof of
Theorem~\ref{thm:polytope_courcelle} but we modify and extend it in a way
that will make it possible to prove the additional properties. In the first phase, 
we construct a polytope
$Q'_c$, an analogue of the polytope $Q_c$ from the aforementioned proof. The vertices of
this polytope correspond to assignments of feasible types to the nodes of the
tree decomposition $T$.
In the second phase, we define a polytope $Q'_\varphi$ as a properly chosen
face of the polytope $Q'_c$, 
analogously to the choice of the face $Q_\varphi$ of the polytope $Q_c$.
The third phase consists only of introducing the variables $y_v^i$ as
a suitable linear combination of $t$ -- this way we obtain the polytope $P$
from the polytope $Q'_\varphi$.

\paragraph{Phase 1: Constructing $Q'_c$}
In the proof of Theorem~\ref{thm:polytope_courcelle}, we obtain the polytope $Q_c$ 
by gluing together polytopes $P_b$, $b\in V(T)$,
in a bottom-up fashion over nodes of a nice tree decomposition of~$G$.
Recall that every $P_b$ is a $0/1$-polytope, has dimension at most $3|\CCtw|$ 
and its number of vertices is at most $|\CCtw|^3$.
Thus, by Corollary~\ref{cor:simplex_idp_tw}, there exists a constructively decomposable
extension $P_b'$ of $P_b$ describable by inequalities of treewidth 
at most $|\CCtw|^3 + 3|\CCtw|=\Oh(|\CCtw|^3)$.

We proceed in the same way as in the construction in the proof of Theorem~\ref{thm:polytope_courcelle} 
but instead of $P_b$, we glue together the polytopes $P'_b$.
At the same time, by Lemma~\ref{lem:tw_glued}
we combine, again in the bottom-up fashion over nodes of the nice tree decomposition $(T, B)$ of~$G$,
the systems of inequalities that describe the polytopes $P_b'$ and also the
tree decompositions of the corresponding Gaifman graphs.
Let $c$ denote the root of the decomposition tree $T$ as in the proof of 
Theorem~\ref{thm:polytope_courcelle}, 
let $D't + C'f = e'$ with $t,f \geq 0$ denote the resulting system of inequalities 
describing the polytope $Q'_c$ and let $(T', B')$ denote the resulting tree decomposition of 
the Gaifman graph $G(D'~C')$.

We shall now prove that the conditions~\eqref{thm:tw_structure} hold for $(T', B')$ by induction.
Then, it will be sufficient in the later stages of the proof to show that they will not be violated.
\begin{lemma}\label{lem:tw_structure_intermediate}
For each $b \in V(T)$, there are matrices $C'_b$ and $D'_b$ such that $C' t_b + D'_b (\bar{t}, \bar{f}) = e'_b$ with $t_b, \bar{t}, \bar{f} \geq 0$ describes the intermediate polytope $Q'_b$ obtained in the bottom-up construction, $G(C'_b~D'_b)$ has a tree decomposition $(T_b, B_b)$ of treewidth $\Oh(|\CCtw|^3)$, $T_b$ is as defined before (i.e., a subtree of $T$ rooted in $b$), and for every node $a \in V(T_b)$, it holds that $\bigcup_{i \in \CCtw} \{t_a[i]\} \subseteq B_b(a)$, and $\bigcup_{i \in \CCtw} \{t_a[i]\} \cap B_b(a') = \emptyset$ for every $a' \not\in \delta_{T_a}(a)$.
\end{lemma}
Clearly, if these conditions hold for the root $c$, the conditions~\eqref{thm:tw_structure} are satisfied.

\begin{proof}
First, let $b$ be a leaf.
By Lemma~\ref{lem:tw_vertices}, there is a system $C'_b t_b + D'_b f_b = e_b$ with $f_b \geq 0$ describing $P'_b$ (which contains just one point), and there is a trivial tree decomposition of $G(C'_b~D'_b)$ with one bag containing all vertices.
This establishes the base case of the induction.

Consider an introduce or forget node $b$ of $T$ with a child $a$, and we glue the polytopes $Q'_a$ and $P'_b$.
By Lemma~\ref{lem:tw_vertices}, $P'_b$ is described by $A_b d_b + C_b t_b + D_b f_b = e_b$ with $f_b \geq 0$, and $G(A_b~C_b~D_b)$ has a trivial tree decomposition with one bag containing all its vertices.
By the induction hypothesis, $Q'_a$ is described by $C'_a t_a + D'_a (\bar{t}, \bar{f}) = e_a$ with $t_a, \bar{t}, \bar{f} \geq 0$, where $\bar{t}$ and $\bar{f}$ are the $t$ and $f$ variables associated with the nodes of $T_a$.
Thus, the following system, where the columns $\begin{smallmatrix}A_b \\ C'_a\end{smallmatrix}$ correspond to the matrix $C'_b$ and the remaining columns to the matrix $D'_b$, describes $Q'_b$:

  \begin{alignat*}{5}
    A_b d_b &~+ &~C_b t_b ~+~ &D_b f_b  &&~\geq e_b \\
    C'_a t_a & &&& ~+~D'_a (\bar{t}, \bar{f}) &~\geq e_a
  \end{alignat*}

Moreover, because there is a tree decomposition $(T_a, B_a)$ of the Gaifman graph $G(A_a~E_a)$ such that $\bigcup_{i \in \CCtw}\{t_a[i]\} \subseteq B_a(a)$,
we are in the situation of the second part of Lemma~\ref{lem:tw_glued} with $A = (A_b~D_b)$, $B = C_b$, $C = C'_a$, $D = D'_a$. 
This implies that $(T_b, B_b)$ with $B_b$ defined by $B_b(a') = B_a(a')$ for all $a' \in T_a$, and $B_b(b) = V(G(A_b~C_b~D_b))$, is a tree decomposition of $G(C'~D')$ which has the desired properties.

The situation is analogous for the join node.
By the first part of Lemma~\ref{lem:tw_glued} together with the induction hypothesis, the treewidth of $(T_b, B_b)$ is clearly at most $\Oh(|\CCtw|^3)$.
\end{proof}
   
\paragraph{Phase 2: Taking the face $Q'_\varphi$}
We take the face $Q'_{\varphi}$ of $Q_c'$
corresponding to the valid inequality $\sum_{t \in \CCtw}
\rho_\varphi(t){\cdot}z_c^t \leq 1$.
That corresponds to adding the equality $\sum_{j \in \CCtw}
\rho_\varphi(t){\cdot}t_c[j] = 1$ to the system $D't + C'f = e'$.
Let us denote $D''t + C''f = e''$ the system obtained from $D't + C'f = e'$ by adding the aforementioned equality.
Adding the new equality corresponds to adding edges to $G(D'~C')$ which are connecting vertices $t_c[j]$.
Since all variables $t_c[j]$ belong to the bag $B'(c)$, the tree decomposition conditions are not violated by adding these edges and $(T', B')$ is a tree decomposition of $G(D''~C'')$ as well.
Thus, the treewidth of $G(D''~C'')$ is the same as the treewidth of $G(D'~C')$.

We will now show that since $Q'_c$ is decomposable and a $0/1$-polytope, $Q'_\varphi$ is decomposable as well.
Let $r \in \NN$ and consider any $x \in rQ'_\varphi$.
Because $x \in rQ'_\varphi \subseteq rQ'_c$, there exist $x^1, \dots, x^r \in Q'_c$ such that $\sum_i x^i = x$.
Because $x \in rQ'_\varphi$, it satisfies $\sum_{t \in \CCtw}
\rho_\varphi(t){\cdot}z_c^t = r$, and because $Q'_\varphi$ is a $0/1$-polytope, every $x^i$ satisfies $\sum_{j \in \CCtw}
\rho_\varphi(t){\cdot}t_c[j] = 1$.
This implies that $x^i \in Q'_\varphi$ for all $i=1,\dots,r$.

\paragraph{Phase 3: Obtaining $P$ by adding variables $y$}
To obtain $P$ from $Q'_\varphi$, it remains
  to add projections $y_v^i=\sum\limits_{j \in \CCtw}
  \mu(t,v,i){\cdot}t_{top(v)}[j]$ for each $v \in V$ and $i \in \{1, \dots, m\}$.
  Now consider the system $Ay + Dt + Cf = e$ which is thus obtained.
  
  Regarding the treewidth of $G(A~D~C)$, note that the sum
  defining each $y_v^i$ only involves variables
  associated with the node $top(v)$.
  Specifically, $Ay + Dt + Cf = e$ can be written as

  \begin{alignat*}{7}
    \mathbf{0}y ~+~ &D''&t ~&+&~ &C''&f & ~&=~ e'' \\
    -Iy ~+~ &\Lambda &t ~&+&~ &\mathbf{0}&f &~&=~ \mathbf{0}
  \end{alignat*}
  where the block $(-I~\Lambda~\mathbf{0})$ corresponds to the projections to $y_v^i$.

  Fix $v \in V$.
  Then in $G(A~D~C)$ the variable $y_v^i$ corresponds to a vertex
  connected to vertices $t_{top(v)}[j]$ for which $\mu(j,v,i)
  = 1$, all of which belong to one bag $B=B'(top(v))$.
  A decomposition $(T^*, B^*)$ of
  $G(A~D~C)$ can be obtained from $(T', B')$ by adding, for every $i \in \{1, \dots, m\}$,
  the vertex corresponding to $y_v^i$ to the bag $B$. This increases
  the width of $B$ by at most $m$.
  Since $top(v)$ is distinct for every $v \in V$, this operation can
  be performed independently for every $v$, resulting in a
  decomposition of width at most $\Oh(|\CCtw|^3) + m = \Oh(|\CCtw|^3)$, satisfying the claimed properties~\eqref{thm:tw_structure}.

  Regarding constructive decomposability, we use Lemma~\ref{lem:idp_projection}. The
  polytope $Q'_\varphi$ is constructively decomposable, and there is a linear projection $\pi$
  with integer coefficients such that $P_\varphi(G) = \pi(Q'_\varphi)$. Thus
  $P = \{(\pi(t), t, f) \mid (t,f) \in Q'_\varphi\}$ is constructively decomposable, satisfying property~\eqref{thm:decomp}.

  Finally, Theorem~\ref{thm:constructive_polytope_courcelle} shows
  that $A, D, C, e$ and $\nu$ can be constructed in the claimed time, satisfying property~\eqref{thm:adcdnu_computable}, and by the definition of $\nu$ and Lemma~\ref{lem:vertices_labelings}, condition~\eqref{thm:variables} is also satisfied, completing the proof.
\end{proof}

As a corollary, we have the following.

\begin{corollary}
  Under the assumptions of Theorem~\ref{thm:polytope_courcelle}, $\xcdec(P_{\varphi}(G)) \leq f(|\varphi|, \tau) \cdot n$.
  \end{corollary}

\subsection{Cliquewidth}

The results of this paper can be extended also to graphs of bounded \emph{clique-width}, a more general class of graphs, at the cost of restricting our logic from MSO$_2$ to MSO$_1$.
A $\gamma$-expression is a concept analogous to a tree decomposition.

\begin{theorem}\label{thm:cw}
Let $G$ be a graph of clique-width $cw(G) = \gamma$ represented as a $\sigma_1$-structure, let $\psi$ be an MSO$_1$ formula with $m$ free variables, and let
$$\displaystyle P_\psi(G) = \conv\left(\left\{y \in \{0,1\}^{nm} \ |\ y
\text{
satisfies } \psi \right\}\right). \ $$

Then $\xc(P_\psi(G)) \leq f(|\psi|,\gamma) \cdot n$ for some computable function $f$.

Moreover, if $G$ is given along with its $\gamma$-expression $\Gamma$, $P_\psi(G)$ can be constructed in time $f(|\psi|, \gamma) \cdot n$.
\end{theorem}

While we could prove Theorem~\ref{thm:cw} directly using notions analogous to feasible types and tuples, it is easier to use a close relationship between cliquewidth and treewidth:

\begin{lemma}[{\cite[Lemma 4.14]{GajarskyHKO:17}}]\label{lem:interpret_cw}
Let $G$ be a graph of clique-width $cw(G) = \gamma$
given along with its $\gamma$-expression~$\Gamma$,
and let $\psi$ be an MSO$_1$ formula.
One can, in time $\Oh(|V(G)|+|\Gamma|+|\psi|)$, compute a tree $T$ and an MSO$_1$
formula $\varphi$ such that
$V(G)\subseteq V(T)$ and
$$\mbox{for every $X$, $T\models\varphi(X)$, \,if and only if\,
	$X\subseteq V(G)$ and $G\models\psi(X)$.}$$
\end{lemma}

\begin{proof}[of Theorem~\ref{thm:cw}]
Let $\Gamma$ be some $\gamma$-expression of $G$.
Consider the tree $T$ of Lemma~\ref{lem:interpret_cw} derived from $G$ and $\gamma$ and the MSO$_1$ formula $\varphi$ derived from $\psi$.
By the relationship between $G$ and $T$ and $\psi$ and $\varphi$ of Lemma~\ref{lem:interpret_cw}, the polytope $P_\varphi(T)$ is an extended formulation of $P_\psi(G)$.
Thus, applying Theorem~\ref{thm:polytope_courcelle} to $T$ and $\varphi$ suffices to show that $\xc(P_\psi(G)) \leq \xc(P_\varphi(T)) \leq f(|\psi|, \gamma) \cdot n$ for some computable function $f$.
If a $\gamma$-expression of $G$ is provided, clearly the proof becomes constructive.
\end{proof}

For simplicity, we have required that $G$ comes along with its $\gamma$-expression to obtain a constructivity in Theorem~\ref{thm:cw}.
This is because it is currently not known how to efficiently construct a $\gamma$-expression
for an input graph of fixed clique-width~$\gamma$.
However, one may instead use the result of
\cite{HlinenyO08} which constructs in FPT time a so-called
rank-decomposition of $G$ which can be used as an approximation of a
$\gamma$-expression for $G$ with up to an exponential jump, but this does
not matter for a fixed parameter~$\gamma$ in theory.

\subsection{Courcelle's Theorem and Optimization.}
It is worth noting that even though linear time {\em optimization}
versions of Courcelle's theorem are known, our result provides a
linear size LP for these problems out of the box. Together with a
polynomial algorithm for solving linear programming we immediately
get the following:

\begin{theorem} \label{thm:courcelle_optimization}
Given a graph $G$ on $n$ vertices with treewidth $\tau$, a formula $\varphi \in
\textrm{MSO}_2$ with $m$ free variables and real weights $w_v^i$, for
every $v \in V(G)$ and $i \in \{1, \dots, m\}$, the problem
\[
  \mbox{opt}  \left\{\sum_{v \in V(G)} \sum_{i=1}^m w_v^i \cdot y_v^i
  \ \bigg| \ y
  \text{ satisfies } \varphi \right\}
\]
where opt is $\min$ or $\max$, is solvable
in time polynomial in the input size.
\end{theorem}
\section{Acknowledgements}
We thank the anonymous reviewers for pointing out existing work, a shorter proof of the Glueing
lemma, and clarifying the distinctions between MSO$_1$ and MSO$_2$, among various other improvements. 

\bibliographystyle{abbrv}
\bibliography{efs_swat}

\begin{thebibliography}{10}

\bibitem{AprileF2019}
M.~Aprile and Y.~Faenza.
\newblock Extended formulations from communication protocols in
  output-efficient time.
\newblock In {\em International Conference on Integer Programming and
  Combinatorial Optimization}, pages 43--56. Springer, 2019.

\bibitem{ALS:91}
S.~Arnborg, J.~Lagergren, and D.~Seese.
\newblock Easy problems for tree-decomposable graphs.
\newblock {\em Journal of Algorithms}, 12(2):308--340, June 1991.

\bibitem{AssadiENYZ:14}
S.~Assadi, E.~Emamjomeh-Zadeh, A.~Norouzi-Fard, S.~Yazdanbod, and
  H.~Zarrabi-Zadeh.
\newblock The minimum vulnerability problem.
\newblock {\em Algorithmica}, 70(4):718--731, 2014.

\bibitem{AT2013}
D.~Avis and H.~R. Tiwary.
\newblock On the extension complexity of combinatorial polytopes.
\newblock In {\em Proc.\ ICALP(1)}, pages 57--68, 2013.

\bibitem{BannachB19}
M.~Bannach and S.~Berndt.
\newblock Practical access to dynamic programming on tree decompositions.
\newblock {\em Algorithms}, 12(8):172, 2019.

\bibitem{BT:78}
S.~Baum and L.~E. Trotter~Jr.
\newblock Integer rounding and polyhedral decomposition for totally unimodular
  systems.
\newblock In {\em Optimization and Operations Research}, pages 15--23.
  Springer, 1978.

\bibitem{BM:15}
D.~{Bienstock} and G.~{Munoz}.
\newblock {LP approximations to mixed-integer polynomial optimization
  problems}.
\newblock {\em ArXiv e-prints}, Jan. 2015.

\bibitem{Bodlaender:93}
H.~L. Bodlaender.
\newblock A linear time algorithm for finding tree-decompositions of small
  treewidth.
\newblock In {\em Proc. STOC}, pages 226--234, 1993.

\bibitem{Bodlaender:06}
H.~L. Bodlaender.
\newblock Treewidth: characterizations, applications, and computations.
\newblock In {\em Proc. of WG}, volume 4271 of {\em LNCS}, pages 1--14.
  Springer, 2006.

\bibitem{BraunFPS15}
G.~Braun, S.~Fiorini, S.~Pokutta, and D.~Steurer.
\newblock Approximation limits of linear programs (beyond hierarchies).
\newblock {\em Math. Oper. Res.}, 40(3):756--772, 2015.

\bibitem{BraunJLP17}
G.~Braun, R.~Jain, T.~Lee, and S.~Pokutta.
\newblock Information-theoretic approximations of the nonnegative rank.
\newblock {\em Computational Complexity}, 26(1):147--197, 2017.

\bibitem{BB:14}
A.~Buchanan and S.~Butenko.
\newblock Tight extended formulations for independent set, 2014.
\newblock Available on Optimization Online.

\bibitem{CCZ:13}
M.~Conforti, G.~Cornu{\'e}jols, and G.~Zambelli.
\newblock Extended formulations in combinatorial optimization.
\newblock {\em Annals of Operations Research}, 204(1):97--143, 2013.

\bibitem{CP12}
M.~Conforti and K.~Pashkovich.
\newblock The projected faces property and polyhedral relations.
\newblock {\em Mathematical Programming}, pages 1--12, 2015.

\bibitem{Courcelle:90}
B.~Courcelle.
\newblock The monadic second-order logic of graphs {I}: Recognizable sets of
  finite graphs.
\newblock {\em Information and Computation}, 85:12--75, 1990.

\bibitem{CMR:98}
B.~Courcelle, J.~A. Makowsky, and U.~Rotics.
\newblock Linear time solvable optimization problems on graphs of bounded
  clique width.
\newblock In {\em Proc. of WG}, volume 1517 of {\em LNCS}, pages 125--150,
  1998.

\bibitem{CM:93}
B.~Courcelle and M.~Mosbah.
\newblock Monadic second-order evaluations on tree-decomposable graphs.
\newblock {\em Theoretical Computer Science}, 109(1--2):49--82, 1~Mar. 1993.

\bibitem{FaenzaFGT15}
Y.~Faenza, S.~Fiorini, R.~Grappe, and H.~R. Tiwary.
\newblock Extended formulations, nonnegative factorizations, and randomized
  com. protocols.
\newblock {\em Math. Program.}, 153(1):75--94, 2015.

\bibitem{FaenzaMP2018}
Y.~Faenza, G.~Mu{\~n}oz, and S.~Pokutta.
\newblock New limits of treewidth-based tractability in optimization.
\newblock {\em arXiv preprint arXiv:1807.02551}, 2018.

\bibitem{FioriniHW2017}
S.~Fiorini, T.~Huynh, and S.~Weltge.
\newblock Strengthening convex relaxations of 0/1-sets using boolean formulas.
\newblock {\em arXiv preprint arXiv:1711.01358}, 2017.

\bibitem{FMPTW15}
S.~Fiorini, S.~Massar, S.~Pokutta, H.~R. Tiwary, and R.~de~Wolf.
\newblock Exponential lower bounds for polytopes in combinatorial optimization.
\newblock {\em J. {ACM}}, 62(2):17, 2015.

\bibitem{Freuder:90}
E.~C. Freuder.
\newblock Complexity of {$K$}-tree structured constraint satisfaction problems.
\newblock In {\em Proc. of the 8th National Conference on Artificial
  Intelligence}, pages 4--9, 1990.

\bibitem{FrickGrohe}
M.~Frick and M.~Grohe.
\newblock The complexity of first-order and monadic second-order logic
  revisited.
\newblock {\em Annals of pure and applied logic}, 130(1-3):3--31, 2004.

\bibitem{GajarskyHT18}
J.~Gajarsk{\'{y}}, P.~Hlinen{\'{y}}, and H.~R. Tiwary.
\newblock Parameterized extension complexity of independent set and related
  problems.
\newblock {\em Discret. Appl. Math.}, 248:56--67, 2018.

\bibitem{GajarskyHKO:17}
J.~{Gajarsk{\'y}}, P.~{Hlin{\v e}n{\'y}}, M.~{Kouteck{\'y}}, and S.~{Onn}.
\newblock {Parameterized Shifted Combinatorial Optimization}.
\newblock {\em ArXiv e-prints}, Feb. 2017.

\bibitem{GPW:07}
G.~Gottlob, R.~Pichler, and F.~Wei.
\newblock Monadic datalog over finite structures with bounded treewidth.
\newblock In {\em Proc. {PODS}}, pages 165--174, 2007.

\bibitem{gruenbaum}
B.~Gr{\"u}nbaum.
\newblock {\em Convex Polytopes}.
\newblock Wiley Interscience Publ., London, 1967.

\bibitem{HlinenyO08}
P.~Hlin\v{e}n{\'{y}} and S.~Oum.
\newblock Finding branch-decompositions and rank-decompositions.
\newblock {\em {SIAM} J. Comput.}, 38(3):1012--1032, 2008.

\bibitem{Kaibel11}
V.~Kaibel.
\newblock Extended formulations in combinatorial optimization.
\newblock {\em Optima}, 85:2--7, 2011.

\bibitem{KaibelL10}
V.~Kaibel and A.~Loos.
\newblock Branched polyhedral systems.
\newblock In {\em Proc. {IPCO}}, volume 6080 of {\em LNCS}, pages 177--190.
  Springer, 2010.

\bibitem{KaibelP11}
V.~Kaibel and K.~Pashkovich.
\newblock Constructing extended formulations from reflection relations.
\newblock In {\em Proc. {IPCO}}, volume 6655 of {\em LNCS}, pages 287--300.
  Springer, 2011.

\bibitem{KLLL:15}
L.~Kaiser, M.~Lang, S.~Le{\ss}enich, and C.~L{\"o}ding.
\newblock {A Unified Approach to Boundedness Properties in MSO}.
\newblock In {\em Proc. of CSL}, volume~41 of {\em LIPIcs}, pages 441--456,
  2015.

\bibitem{Kloks:94}
T.~Kloks.
\newblock {\em Treewidth: Computations and Approximations}, volume 842 of {\em
  LNCS}.
\newblock Springer, 1994.

\bibitem{KLR:11}
J.~Kneis, A.~Langer, and P.~Rossmanith.
\newblock Courcelle's theorem - {A} game-theoretic approach.
\newblock {\em Discrete Optimization}, 8(4):568--594, 2011.

\bibitem{KnopKMT:17}
D.~Knop, M.~Kouteck{\'y}, T.~Masa{\v{r}}{\'i}k, and T.~Toufar.
\newblock Simplified algorithmic metatheorems beyond {MSO}: Treewidth and
  neighborhood diversity, Mar.~1 2017.
\newblock Accepted to WG 2017.

\bibitem{KV:98}
P.~G. Kolaitis and M.~Y. Vardi.
\newblock Conjunctive-query containment and constraint satisfaction.
\newblock In {\em Proc. {PODS}}, 1998.

\bibitem{KK:15}
P.~Kolman and M.~Kouteck{\'y}.
\newblock Extended formulation for {CSP} that is compact for instances of
  bounded treewidth.
\newblock {\em Electr. J. Comb}, 22(4):P4.30, 2015.

\bibitem{Koutecky:2017}
M.~Koutecky.
\newblock {\em {Treewidth, Extended Formulations of CSP and MSO Polytopes, and
  their Algorithmic Applications}}.
\newblock PhD thesis, Charles University, 2017.

\bibitem{Kreutzer:08}
S.~Kreutzer.
\newblock Algorithmic meta-theorems.
\newblock In {\em Proc. of {IWPEC}}, volume 5018 of {\em LNCS}, pages 10--12.
  Springer, 2008.

\bibitem{LangerRRS12}
A.~Langer, F.~Reidl, P.~Rossmanith, and S.~Sikdar.
\newblock Evaluation of an mso-solver.
\newblock In D.~A. Bader and P.~Mutzel, editors, {\em Proceedings of the 14th
  Meeting on Algorithm Engineering {\&} Experiments, {ALENEX} 2012, The Westin
  Miyako, Kyoto, Japan, January 16, 2012}, pages 55--63. {SIAM} / Omnipress,
  2012.

\bibitem{LRRS:14}
A.~Langer, F.~Reidl, P.~Rossmanith, and S.~Sikdar.
\newblock Practical algorithms for {MSO} model-checking on tree-decomposable
  graphs.
\newblock {\em Computer Science Review}, 13-14:39--74, 2014.

\bibitem{Laurent:09}
M.~Laurent.
\newblock Sums of squares, moment matrices and optimization over polynomials.
\newblock In {\em Emerging applications of algebraic geometry}, pages 157--270.
  Springer, 2009.

\bibitem{LeeRS15}
J.~R. Lee, P.~Raghavendra, and D.~Steurer.
\newblock Lower bounds on the size of semidefinite programming relaxations.
\newblock In {\em Proc. {STOC}}, pages 567--576, 2015.

\bibitem{Libkin:04}
L.~Libkin.
\newblock {\em Elements of Finite Model Theory}.
\newblock Springer-Verlag, Berlin, 2004.

\bibitem{Margot_thesis}
F.~Margot.
\newblock {\em Composition de polytopes combinatoires: une approche par
  projection}.
\newblock PhD thesis, {\'E}cole polytechnique f{\'e}d{\'e}rale de Lausanne,
  1994.

\bibitem{Martin:1990}
R.~K. Martin, R.~L. Rardin, and B.~A. Campbell.
\newblock Polyhedral characterization of discrete dynamic programming.
\newblock {\em Oper. Res.}, 38(1):127--138, Feb. 1990.

\bibitem{OmranSZ:2013}
M.~T. Omran, J.-R. Sack, and H.~Zarrabi-Zadeh.
\newblock Finding paths with minimum shared edges.
\newblock {\em J. Comb. Optim}, 26(4):709--722, 2013.

\bibitem{Rosenthal:1973}
R.~W. Rosenthal.
\newblock A class of games possessing pure-strategy nash equilibria.
\newblock {\em International Journal of Game Theory}, 2(1):65--67, 1973.

\bibitem{Schrijver86}
A.~Schrijver.
\newblock {\em Theory of Linear and Integer Programming}.
\newblock Wiley-Interscience Series in Discrete Mathematics. John Wiley \&
  Sons, 1986.

\bibitem{Sellmann:08}
M.~Sellmann.
\newblock The polytope of tree-structured binary constraint satisfaction
  problems.
\newblock In {\em Proc. {CPAIOR}}, volume 5015 of {\em LNCS}, pages 367--371.
  Springer, 2008.

\bibitem{SML:07}
M.~Sellmann, L.~Mercier, and D.~H. Leventhal.
\newblock The linear programming polytope of binary constraint problems with
  bounded tree-width.
\newblock In {\em Proc. {CPAIOR}}, volume 4510 of {\em LNCS}, pages 275--287.
  Springer, 2007.

\bibitem{VanderbeckWolsey2010}
F.~Vanderbeck and L.~A. Wolsey.
\newblock Reformulation and decomposition of integer programs.
\newblock In {\em 50 {Y}ears of {I}nteger {P}rogramming 1958-2008}, pages
  431--502. Springer, 2010.

\bibitem{Wolsey11}
L.~A. Wolsey.
\newblock Using extended formulations in practice.
\newblock {\em Optima}, 85:7--9, 2011.

\bibitem{Yannakakis91}
M.~Yannakakis.
\newblock Expressing combinatorial optimization problems by linear programs.
\newblock {\em J. Comput. Syst. Sci.}, 43(3):441--466, 1991.

\bibitem{ziegler}
G.~M. Ziegler.
\newblock {\em Lectures on Polytopes}, volume 152 of {\em Graduate Texts in
  Mathematics}.
\newblock Springer-Verlag, 1995.

\end{thebibliography}

\clearpage

\appendix

\section{Definitions of MSO$_1$ and MSO$_2$ formulae}

\begin{definition}[MSO formulae over a vocabulary] 
Given a vocabulary $\sigma$ consisting of relation symbols $P_1, P_2,\ldots$ of
arities $r_1, r_2, \ldots$ and constants $c_1, c_2, \dots$, the set of MSO$[\sigma]$-formulae is the smallest set of formulae 
such that 
\begin{enumerate}
\item For every relation symbol $P_i$ and any $r_i$-tuple $y_1, \dots, y_{r_i}$ where each $y_j$ is either a first-order formula or a constant,
$R(y_1, \ldots, y_{r_i})$ is an MSO$[\sigma]$ formula.
\item $y\in X$  is an MSO$[\sigma]$ formula for every first order variable $y$ and every constant $y$ and every second-order variable $X$.
\item If $\phi_1$ and $\phi_2$ are MSO$[\sigma]$ formulae then $\phi_1 \wedge \phi_2$, $\phi_1 \vee \phi_2$ and $\neg \phi_1$ are MSO$[\sigma]$ formulae.
\item If $\phi(x)$ is an MSO$[\sigma]$ formula and $x$ is a first-order variable then $\exists x \phi(x)$ and $\forall x \phi(x)$ are MSO$[\sigma]$ formulae. 
\item If $\phi(X)$ is an MSO$[\sigma]$ formula and $X$ is a second-order variable then $\exists X \phi(X)$ and $\forall X \phi(X)$ are MSO$[\sigma]$ formulae.   
\end{enumerate}
\end{definition}

MSO$_1$ is the set of MSO formulae over the vocabulary $\sigma_1$ and MSO$_2$ is the
set of MSO formulae over the vocabulary $\sigma_2$; recall that $\sigma_1$
is the vocabulary consisting of a single binary predicate $E$, and $\sigma_2$ is the 
vocabulary consisting of two unary predicates $L_V$ and $L_E$ and a binary
predicate $E_I$.

\end{document}